\documentclass[10pt]{article}

\usepackage[authoryear,round]{natbib}            
\usepackage{amsmath}                             
\numberwithin{equation}{section}
\usepackage{graphicx}                           
\usepackage{subfigure}
\usepackage[T1]{fontenc}
\usepackage{lmodern} 
\usepackage{enumitem}                           
\usepackage{hyperref}
\usepackage{amssymb}                            
\usepackage[mathscr]{eucal}                     
\usepackage{dsfont}
\usepackage{rotating}
\usepackage{cancel}
\usepackage{lscape}
\usepackage[paperwidth=8.5in,paperheight=11in,top=1.25in, bottom=1.25in, left=1.00in, right=1.00in]{geometry}
\usepackage{mathtools}                
\mathtoolsset{showonlyrefs=true}
\usepackage{fixltx2e,amsmath}                                           
\MakeRobust{\eqref}

\usepackage{mathdots}
\usepackage{amsthm}                                                             
\allowdisplaybreaks                                                             
\theoremstyle{plain}
\newtheorem{theorem}{Theorem}[section]

\newtheorem{lemma}[theorem]{Lemma}                              
\newtheorem{proposition}[theorem]{Proposition}

\theoremstyle{definition}
\newtheorem{definition}[theorem]{Definition}

\newtheorem{remark}[theorem]{Remark}
\newtheorem{assumption}[theorem]{Assumption}

\def \b {{\beta}}

\def \b {{\beta}}
\def \d {{\delta}}

\def \G {{\Gamma}}
\def \Gt {{\tilde{\Gamma}}}

\def \sig {{\sigma}}

\def \R {{\mathds {R}}}
\def \N {{\mathds {N}}}

\def \e {{\varepsilon}}
\def \eps {{\varepsilon}}

\def \tilde {\widetilde}

\def\p{\partial}

\newcommand{\<}{\langle}
\renewcommand{\>}{\rangle}
\renewcommand{\(}{\left(}
\renewcommand{\)}{\right)}
\renewcommand{\[}{\left[}
\renewcommand{\]}{\right]}

\newcommand\Eb{\mathds{E}}

\newcommand\Qb{\mathds{Q}}
\newcommand\Rb{\mathds{R}}

\newcommand\Ac{\mathscr{A}}
\newcommand\Act{\tilde{\mathscr{A}}}

\newcommand\Bc{\mathscr{B}}

\newcommand\Fc{\mathscr{F}}
\newcommand\Gc{\mathscr{G}}

\newcommand\Lc{\mathscr{L}}
\newcommand\Mc{\mathscr{M}}
\newcommand\Nc{\mathscr{N}}
\newcommand\Oc{\mathscr{O}}
\newcommand\Pc{\mathscr{P}}

\newcommand\Xc{\mathscr{X}}
\newcommand\Xct{\tilde{\mathscr{X}}}
\newcommand\Vc{\mathscr{V}}

\renewcommand\phi{\varphi}

\newcommand\Om{\Omega}

\newcommand\gam{\gamma}
\newcommand\Gam{\Gamma}
\newcommand\lam{\lambda}
\newcommand\del{\delta}

\newcommand\Chi{\mathcal{X}}

\newcommand\xb{\bar{x}}
\newcommand\yb{\bar{y}}

\newcommand\sigb{\bar{\sig}}

\newcommand\ub{\bar{u}}

\newcommand\Cv{\mathbf{C}}
\newcommand\mv{\mathbf{m}}

\newcommand\phih{\widehat{\phi}}
\newcommand\hh{\widehat{h}}

\newcommand\Vct{\widetilde{\Vc}}

\newcommand\Xt{\widetilde{X}}
\newcommand\xt{\widetilde{x}}
\newcommand\zt{\widetilde{z}}
\newcommand\Zt{\widetilde{Z}}
\newcommand\ut{\widetilde{u}}
\newcommand\at{\widetilde{a}}
\newcommand\bt{\widetilde{b}}
\newcommand\ft{\widetilde{f}}
\newcommand\ct{\widetilde{c}}

\renewcommand\d{\partial}

\newcommand\ii{\mathtt{i}}
\newcommand\dd{\mathrm{d}}
\newcommand\ee{\mathrm{e}}
\newcommand\BS{\text{\rm BS}}

\begin{document}

\title{Leveraged ETF implied volatilities from ETF dynamics}

\author{
Tim Leung
\thanks{Industrial Engineering \& Operations Research Department, Columbia University, New York, NY 10027.
\textbf{E-mail:} \url{leung@ieor.columbia.edu}.}
\and
Matthew Lorig
\thanks{Department of Applied Mathematics, University of Washington, Seattle, WA 98195.
\textbf{E-mail:} \url{mattlorig@gmail.com}.
Work partially supported by NSF grant DMS-0739195.}
\and
Andrea Pascucci
\thanks{
Dipartimento di Matematica,
Universit\`a di Bologna, Bologna, Italy.
\textbf{E-mail:} \url{andrea.pascucci@unibo.it}.}
}

\date{This version: \today}

\maketitle

\begin{abstract}
The growth of the exchange-traded fund (ETF) industry has given rise to the trading of options written on ETFs and their leveraged counterparts {(LETFs)}. We study the relationship between the ETF and LETF implied volatility surfaces when the underlying ETF is modeled by a general class of local-stochastic volatility models.  A closed-form approximation for prices is derived for European-style options whose payoff depends on the terminal value of the ETF and/or LETF.  Rigorous error bounds for this pricing approximation are established.  A closed-form approximation for implied volatilities is also derived.   We also discuss a scaling procedure for comparing implied volatilities across leverage ratios.  The implied volatility expansions and scalings are tested in three well-known settings: CEV, Heston and SABR.
\end{abstract}

\noindent
\textbf{Keywords}:  implied volatility, local-stochastic volatility, leveraged exchange-traded fund, implied volatility scaling

%
%

\section{Introduction}
\label{sec:intro}
The market of exchange-traded funds (ETFs) has been growing at a robust pace since their introduction in 1993\footnote{The first US-listed ETF, the SPDR S\&P 500 ETF (SPY),  was launched on  January 29th,  1993.}. As of the end of 2012, the global ETF industry has over
\$1.8 trillion in assets under management (AUM) comprised  of 4,272 products, and has seen close to \$200 billion of positive capital inflows\footnote{``2013 ETF \& Investment Outlook" by David Mazza, SPDR  ETF Strategy \& Consulting, State
Street Global Advisors. Available at \url{http://www.spdr-etfs.com}.}. In recent years, a sub-class of  ETFs, called {leveraged} ETFs (LETFs),  has gained popularity among investors for their accessibility and liquidity for leveraged positions. These funds are  designed to replicate multiples of the daily returns of some reference index or asset.  For instance, the ProShares S\&P 500 Ultra  (SSO) and UltraPro   (UPRO) are advertised to generate, respectively, 2 and 3 times of the daily
returns of  the S\&P 500 index, minus a small expense fee.  On the other hand, an LETF with a negative leverage ratio   allows investors to take a bearish position on the underlying index  by longing the fund.   An example is the ProShares S\&P 500 UltraShort  (SDS)   with leverage ratio of $-2$. The most   typical  leverage ratios are  $\{-3, -2, -1, 2, 3\}$. With the same reference, such as the S\&P 500, these LETFs share  very similar sources of randomness, but they also exhibit different path behaviors (see \cite{cheng} and \cite{avellaneda1}).

The use of ETFs has also led to increased trading of  options  written on ETFs.  During  2012, the total  options contract volume traded at    Chicago Board Options   Exchange  (CBOE)  is 1.06 billion contracts, of which  282 million contracts  are ETF options while 473 million are equity options.   This leads to an  important  question of consistent pricing of options on ETFs and LETFs with the same reference. Since options  are commonly quoted and compared in terms of implied volatility, it is natural to consider the implied volatility relationships among LETF options, not only across strikes and maturities, but also for various leverage ratios.

In this paper, we analyze the implied volatility surfaces associated with European-style LETF options in a general class of local-stochastic volatility (LSV) models.  {Our approach is to (i) find an expansion for approximate LETF option prices (ii) establish rigorous error bounds for this approximation and (iii) translate the price approximation into approximate implied volatilities.  Exact pricing and implied volatility formulas in a general LSV setting are obviously impossible to obtain.  There are a number of approaches one could feasibly take in order to approximate European-style option prices and their associated implied volatilities.
We review some recent approaches for unleveraged products here.
\cite{Gatheral2012} use heat kernel methods in a local volatility setting.
\cite{benhamou2010time} use a small volatility of volatility expansion for the time-dependent Heston model.  More recently, \cite{bompisgobet2013} use Malliavin calculus to obtain approximations in a quite general LSV setting.
And, \cite{forde-jacquier-uncorrelated} use the Freidlin-Wentzell theory of large deviations to analyze an uncorrelated LSV model.}

In this paper, we use a polynomial operator expansion technique to obtain approximate prices and implied volatilities.
{The polynomial operator expansion technique was first introduced in
\cite{Pagliarani2012} and \cite{Pagliarani2013a} to compute option prices in a scalar jump-diffusion setting.  It was further developed in \cite{lorig-pagliarani-pascucci-2} to obtain approximate prices and implied volatilities in a multidimensional local-stochastic volatility setting (see also \cite{lorig-pagliarani-pascucci-1} for pricing approximations for models with jumps).}
The reason for basing our expansions on the methods developed in \cite{lorig-pagliarani-pascucci-2} is
that these methods allow us to consider a large class of LSV models for the ETF; many
of the above mentioned methods work only for specific ETF dynamics.
However, without further development, the methods described in \cite{lorig-pagliarani-pascucci-2} are not
sufficient for the rigorous error bounds we establish in this paper.  Indeed, in
\cite{lorig-pagliarani-pascucci-2}, error bounds are established under a uniform ellipticity
assumption.  As we shall see, the generator of the joint ETF/LETF process is not elliptic.  As
such, to establish rigorous error bounds for LETF option prices, we must work in this challenging
non-elliptic setting.

Perhaps the most useful  result of our analysis is the general expression we obtain for the implied volatility expansion.  This expansion allows us to pinpoint the non-trivial  role played by the leverage ratio $\beta$, and thus,  relate the  implied volatility surfaces  between  (unleveraged) ETF and LETF options.  This also  motivates us to apply the idea of \textit{log-moneyness scaling}, with the  objective  to  view the implied volatilities across leverage ratios on the same scale and orientation. In particular, for a negative leverage ratio and up to the first-order in log-moneyness,  the LETF implied volatility  is known to be upward sloping while the ETF and long-LETF implied volatilities are downward sloping (see \cite{leungsircarLETF}). The scaling is capable of appropriately adjusting the level and shape of the  implied volatility so that the ETF and LETF implied volatilities  match closely under a given  model. For illustration, we test our implied volatility expansions and the log-moneyness scaling  in three well-known settings: CEV, Heston and SABR, and find that they are very accurate.

{In the recent paper,  \cite{leungsircarLETF} apply  asymptotic techniques to understand the link between implied volatilities of the  ETF and  LETFs of different leverage ratios within a multiscale stochastic volatility framework (see \cite{fpss} for a review of multiscale methods). They also introduce implied volatility scaling procedure, different from our own, in order to identify  possible price discrepancies  in the ETF and LETF options markets. In contrast to their work, the current paper studies the problem in a general LSV framework, which naturally includes well-known models such as CEV, Heston and SABR models, among others.  Moreover, while \cite{leungsircarLETF} obtain an implied volatility approximation that is linear in log-moneyness, we
provide a general expression for LETF implied volatilities that is quadratic in log-moneyness.  We also provide formulas for three specific models (CEV, Heston and SABR) that are cubic in log-moneyness.}
\par
{\cite{haughLETF} propose a heuristic  approximation to compute LETF option prices with Heston stochastic volatility and jumps for the underlying.}  While they do not investigate the implied volatilities, they point out that if the underlying ETF admits the Heston (no jumps) dynamics, then  the LETF also has Heston dynamics with different parameters.  As a particular example of LSV models, {we also obtain the same result revealed through our implied volatility expansions (see Section \ref{sec:heston}).

\par
The rest of this paper proceeds as follows.  In Section \ref{sec:model} we review how LETF dynamics are related to ETF dynamics in a general diffusion setting.  We then introduce Markov dynamics for a general class of LSV models for the ETF.  Next, in Section \ref{sec:pricing}, we formally construct an asymptotic expansion for European-style options whose payoff depends on the terminal value of the ETF and/or LETF.  Rigorous error bounds for our pricing approximation are established in Section \ref{sec:accuracy}.  In Section \ref{sec:imp.vol} we translate our asymptotic expansion for prices into an asymptotic expansion for implied volatilities.  We also discuss some natural scalings of the implied volatility surface of the LETF.  Finally, in Section \ref{sec:examples} we implement our implied volatility expansion in three well-known settings: CEV, Heston and SABR.  Some concluding remarks are given in Section \ref{sec:conclusion}.

%
%

\section{Leveraged ETF dynamics}
\label{sec:model}
We take as given an equivalent martingale measure $\Qb$, chosen by the market on a complete filtered probability space $(\Om,\Fc,\{\Fc_t,t\geq0\},\Qb)$.  The filtration $\{\Fc_t,t\geq0\}$ represents the history of the market.  All stochastic processes defined below live on this probability space and all expectations are taken with respect to $\Qb$.
For simplicity, we assume a frictionless market, no arbitrage, zero interest
rates and no dividends.  We will discuss how to relax these assumptions in Remark \ref{rmk:interest}.
\par
Let $S$ be the price process of an Exchange-Traded Fund (ETF).  We assume $S$ can be modeled under $\Qb$ as a strictly positive It\^o diffusion.  Specifically, we have
\begin{align}
\text{ETF}:&&
S_t
    &=  \ee^{ X_t } , &
\dd X_t
   &=  -\frac{1}{2} \sig_t^2 \dd t + \sig_t \, \dd W_t^x , \label{eq:S}
\end{align}
where $\sig$ is a strictly positive stochastic process.
Note that the drift is fixed by the volatility so that $S$ is a martingale.  Let $L$ be the price process of a Leveraged Exchange-Traded Fund (LETF) with underlying $S$ and with leverage ratio $\beta$.  Typical values of $\beta$ are $\{-3,-2,-1,2,3\}$.
The LETF is
managed as follows: for every unit of currency a trader invests in $L$, the LETF manager borrows $(\beta-1)$ units of currency and invests $\beta$ units of currency in $S$.
The fund manager also typically charges the trader a small expense rate, which, for simplicity, we
assume is zero. Then the dynamics of $L$ are related to $S$ as follows
\begin{align}
\frac{\dd L_t}{L_t}
    =  \beta \, \frac{ \dd S_t}{S_t} =\beta \sig_t \, \dd W_t^x,
\end{align}
and thus we have
\begin{align}
\text{LETF}:&&
L_t
    &=  \ee^{Z_t} , &
\dd Z_t
  &=  -\frac{1}{2} \beta^2 \sig_t^2 \, \dd t + \beta \sig_t \, \dd W_t^x . \label{eq:L}
\end{align}
Comparing \eqref{eq:S} with \eqref{eq:L}, we observe that the volatility of $L$ is scaled by a factor of $\beta$. Moreover, {as shown by \cite{avellaneda1}, one can solve explicitly the SDE for $Z$ in order to obtain an expression for $Z_t$ in terms of $X_t$ and the quadratic variation (integrated variance) of $X$ up to time $t$.  Specifically, we have
\begin{align}\label{eq:ZinX}
 Z_t - Z_0
    &=  \beta \, (X_t- X_0) - \frac{\beta ( \beta - 1 )}{2}  \int_0^t\sig_s^{2} \, \dd s.
\end{align}
Equation \eqref{eq:ZinX} shows that the $\log$ returns of an LETF is the sum of two terms.  The first term is proportional to the $\log$ returns of the underlying ETF.  The second term is proportional to the integrated variance of $X$, and highlights the fact that options on LETFs are \emph{path dependent} options.}
Note that, for leverage ratio $\beta\in\{-3, -2, -1, 2, 3\}$, the coefficient $- \frac{\beta ( \beta - 1 )}{2}$ of the realized variance is strictly negative.  Note also that $- \frac{\beta ( \beta - 1 )}{2}$ is an asymmetric function of $\beta$.


\subsection{Local-stochastic volatility framework}
We now specialize to the Markov setting.  We introduce an auxiliary process $Y$, which is intended to capture effects such as stochastic volatility.  We assume that the triple $(X,Y,Z)$ can be modeled by the following Stochastic Differential Equation (SDE):
\begin{align}
\begin{aligned}
\dd X_t
   &=  -\frac{1}{2} \sig^2(t,X_t,Y_t) \dd t + \sig(t,X_t,Y_t) \dd W_t^x , \\
\dd Y_t
    &=  c(t,X_t,Y_t) \dd t + g(t,X_t,Y_t) \dd W_t^y , \\
\dd Z_t
  &=  -\frac{1}{2} \beta^2 \sig^2(t,X_t,Y_t) \dd t
            + \beta \sig(t,X_t,Y_t) \dd W_t^x , \\
\dd \< W^x, W^y \>_t
    &=  \rho(t,X_t,Y_t) \dd t .
\end{aligned} \label{eq:SDE.0}
\end{align}
We assume that SDE \eqref{eq:SDE.0} has a unique strong solution and that the coefficients $(\sig,c,\rho)$ are smooth.  Sufficient conditions for a unique strong solution are given in \cite{Ikedabook}.  The class of models described by \eqref{eq:SDE.0} enjoys the following features:
\begin{enumerate}
\item \textbf{Stochastic Volatility}: When $\sig$ and $\rho$ are functions of $(t,y)$ only (as they would be in a stochastic volatility model such as Heston), then the pairs $(X,Y)$ and $(Y,Z)$ are Markov processes.  From a mathematical point of view, the lack of $x$-dependence in the correlation $\rho$ and volatility $\sig$ greatly simplifies the pricing and implied volatility analysis, since calls written on $Z$ can be analyzed independently from calls on $X$.
\item \textbf{Local Volatility}: If both $\sig$ and $\rho$ are dependent on $(t,x)$ only (as they would be in a local volatility model such as CEV), then $X$ alone and the pair $(X,Z)$ are Markov processes.  In this case, calls on $X$ can be analyzed separately from $Z$.  However, calls on $Z$ must be analyzed in conjunction with $X$.
\item \textbf{Local-Stochastic Volatility}: If $\sig$ and $\rho$ depend on  $(x,y)$ (as would be the case in a local-stochastic volatility setting such as SABR), then the pair $(X,Y)$ is a Markov process, as is the triple $(X,Y,Z)$.  In this case, options on $X$ can be analyzed independently from $Z$.  In contrast, to analyze options on $Z$, one must consider the triple $(X,Y,Z)$.
\item If $\beta=1$,
then from \eqref{eq:SDE.0} we see that $\dd X_t = \dd Z_t$.  Thus, we need only to obtain prices and implied volatilities for options written on $Z$.  Options written on $X$ can always be obtained by considering the special case $\beta=1$.
\end{enumerate}

%
%

\section{Option pricing}
\label{sec:pricing}
Using risk-neutral pricing and the Markov property of the process $(X,Y,Z)$, we can write the time $t$ price of an option $u(t,x,y,z)$ with expiration date $T>t$ and payoff $\phi(Z_T)$ as the risk-neutral expectation of the payoff
\begin{align}\label{and1}
u(t,x,y,z)
    &=  \Eb[ \phi(Z_T) | X_t = x, Y_t = y, Z_t = z ] .
\end{align}
Under mild assumptions, the function $u$ satisfies the Kolmogorov backward equation
\begin{align}
\( \d_t + \Ac(t) \) u(t)
    &=  0 , &
u(T)
    &=  \phi , \label{eq:u.pde}
\end{align}
where the operator $\Ac(t)$ is given by
\begin{align}
\Ac(t)
    &=  a(t,x,y) \( \( \d_x^2 - \d_x \) + \beta^2 \( \d_z^2 - \d_z \) + 2 \beta \, \d_x \d_z \) \\ & \qquad
            + b(t,x,y) \d_y^2 + c(t,x,y) \d_y + f(t,x,y) \( \d_x \d_y + \beta \, \d_y \d_z \) ,
            \label{eq:A}
\end{align}
with the functions $(a,b,f)$ defined as
\begin{align}
a(t,x,y)
    &=  \tfrac{1}{2} \sig^2(t,x,y) , &
b(t,x,y)
    &=  \tfrac{1}{2} g^2(t,x,y) , &
f(t,x,y)
    &=  g(t,x,y) \sig(t,x,y) \rho(t,x,y) .
\end{align}
For general $(a,b,c,f)$, an explicit solution to \eqref{eq:u.pde} is not available.  Thus, our goal is to find a closed form approximation for the option price $u$ and derive rigorous error bounds for our approximation.
\begin{remark}
\label{rmk:elliptic}
We note that the matrix of second-order derivatives of $\Ac(t)$
\begin{align}
\frac{1}{2}
\begin{pmatrix}
2 a & f & 2 \beta a \\
f & 2 b & \beta f \\
2 \beta a & \beta f & 2 \beta^2 a
\end{pmatrix} \label{eq:matrix}
\end{align}
is singular; the eigenvector $(\beta,0,-1)$ corresponds to eigenvalue zero.  Therefore, the operator $\Ac(t)$ is \emph{not} elliptic.
This gives rise to an additional mathematical challenge in establishing error bounds for the pricing approximation, which we will carry out in Section \ref{sec:accuracy}.
\end{remark}
\begin{remark}[Deterministic interest rates, dividends and expense ratios]
\label{rmk:interest}
Suppose interest rates are a deterministic function $r(t)$ of time.  Suppose also that the ETF holder receives a dividend $q(t) S_t$ per unit time, and the LETF provider charges an expense rate $c(t) L_t$ per unit time where $q(t)$ and $c(t)$ are deterministic functions.  In this case options prices are computed as discounted expectations of the form
\begin{align}
\ut(t,\xt,y,\zt)
    &:= \Eb[ \ee^{-\int_t^T \dd s \,  r(s)} \phi(\Zt_T) | \Xt_t = \xt , Y_t = y, \Zt_t = z] , \\
\dd \Xt_t
    &=  \dd X_t + \( r(t) - q(t) \) \dd t , \\
\dd \Zt_t
    &=  \dd Z_t + \( r(t) - c(t) - \beta q(t) \) \dd t ,
\end{align}
with $(X,Y,Z)$ as given in \eqref{eq:SDE.0}.  Upon making the following change of variables
\begin{align}
u(t,x(t,\xt),y,z(t,\zt))
    &:= \ee^{\int_t^T \dd s \,  r(s)} \ut(t,\xt,y,\zt) , \label{eq:u.2} \\
x(t,\xt)
    &:= \xt + \int_t^T \dd s \, r(s) , \\
z(t,\zt)
    &:= \zt + \int_t^T \dd s \, \( r(s) - c(s) - \beta q(s) \),
\end{align}
a simple application of the chain rule reveals that $u$ as defined in \eqref{eq:u.2} satisfies Cauchy problem \eqref{eq:u.pde}.  Thus, the current framework allows us to readily accommodate these additional features.
\end{remark}

\subsection{Asymptotic prices via Taylor and Dyson series}
\label{sec:dyson}
In this section, we show how Taylor and Dyson series can be combined in order to formally
construct an asymptotic approximation of the solution $u$ of Cauchy problem \eqref{eq:u.pde}.
Throughout the derivation that follows we assume that for every $t$ the coefficients $(a,b,c,f)$
of the operator $\Ac(t)$ are analytic in $(x,y)$ so we can expand each of these functions as a
Taylor series.  As we will see, this assumption is not necessary for the $N$th-order approximation
of $u$, which we will give in Definition \ref{def:ub.N}.  However, making this assumption will
simplify the derivation that follows.
\par
Let $(\xb(\cdot),\yb(\cdot)): [0,T] \to \Rb^2$ be a piecewise continuous map. For any $(t,x,y)$ we
have:
\begin{align}
\chi(t,x,y)
    &=  \sum_{n=0}^\infty \sum_{k=0}^n \chi_{n-k,k}(t) (x-\xb(t))^{n-k} (y-\yb(t))^{k} , \\
\chi_{n-k,k}(t)
    &=  \frac{\d_x^{n-k}\d_y^{k}\chi(t,\xb(t),\yb(t))}{(n-k)!k!} , &
 \chi
    &\in  \{a,b,c,f\} .
\end{align}
Formally, the operator $\Ac(t)$ can now be written as
\begin{align}
\Ac(t)
    &=  \Ac_0(t) + \Bc_1(t) , &
\Bc_1(t)
    &=  \sum_{n=1}^\infty \Ac_n(t) , &
\Ac_n(t)
    &=  \sum_{k=0}^n { (x-\xb(t))^{n-k} (y-\yb(t))^{k} } \Ac_{n-k,k}(t) , \label{eq:A.expand}
\end{align}
where
\begin{align}
\Ac_{n-k,k}(t)
    &=  a_{n-k,k}(t) \( \( \d_x^2 - \d_x \) + \beta^2 \( \d_z^2 - \d_z \) + 2 \beta \, \d_x \d_z \) \\ &\qquad
            + b_{n-k,k}(t) \d_y^2 + c_{n-k,k}(t) \d_y + f_{n-k,k}(t) \( \d_x \d_y + \beta \, \d_y \d_z \) , \label{eq:A.nk}
\end{align}
Inserting expansion \eqref{eq:A.expand} for $\Ac(t)$ into Cauchy problem \eqref{eq:u.pde} we find
\begin{align}
( \d_t + \Ac_0(t) ) u(t)
    &=  - \Bc_1(t) u(t) , &
u(T)
    &=  \phi .
\end{align}
By construction, the operator $\Ac_0(t)$ is the generator of a diffusion with
coefficients that are deterministic functions of time only.  By Duhamel's principle, we therefore
have
\begin{align}
u(t)
    &=  \Pc_0(t,T) \phi + \int_t^T \dd t_1 \, \Pc_0(t,t_1) \Bc_1(t_1) u(t_1) , \label{eq:u.duhemel}
\end{align}
where $\Pc_0(t,T) = \exp \int_t^T \dd s \, \Ac_0(s)$, is the \emph{semigroup} of operators generated by $\Ac_0(t)$; we will provide an explicit form for $\Pc_0(t,T)$ in Section \ref{sec:u0}.
Inserting expression \eqref{eq:u.duhemel} for $u$ back in to the right-hand side of \eqref{eq:u.duhemel} and iterating we obtain
\begin{align}
u(t)
    &=  \Pc_0(t,T) \phi + \sum_{k=1}^\infty
            \int_{t}^T \dd t_1 \int_{t_1}^T \dd t_2 \cdots \int_{t_{k-1}}^T \dd t_k
            \\ & \qquad
            \Pc_0(t,t_1) \Bc_1(t_1)
            \Pc_0(t_1,t_2) \Bc_1(t_2) \cdots
            \Pc_0(t_{k-1},t_k) \Bc_1(t_k)
            \Pc_0(t_k,T) \phi
            \label{eq:dyson} \\
    &=  \Pc_0(t,T) \phi + \sum_{n=1}^\infty \sum_{k=1}^n
            \int_{t}^T \dd t_1 \int_{t_1}^T \dd t_2 \cdots \int_{t_{k-1}}^T \dd t_k
            \\ & \qquad \sum_{i \in I_{n,k}}
            \Pc_0(t,t_1) \Ac_{i_1}(t_1)
            \Pc_0(t_1,t_2) \Ac_{i_2}(t_2) \cdots
            \Pc_0(t_{k-1},t_k) \Ac_{i_k}(t_k)
            \Pc_0(t_k,T) \phi ,
            \label{eq:u.dyson} \\
I_{n,k}
    &= \{ i = (i_1, i_2, \cdots , i_k ) \in \mathds{N}^k : i_1 + i_2 + \cdots + i_k = n \}.
            \label{eq:Ink}
\end{align}
Note that the second to last equality \eqref{eq:dyson} is the classical Dyson series expansion of $u$ corresponding to order zero generator $\Ac_0(t)$ and perturbation $\Bc_1(t)$.  To obtain \eqref{eq:u.dyson} from \eqref{eq:dyson} we have used the fact that, by \eqref{eq:A.expand}, the operator $\Bc_1(t)$ is an infinite sum.
{Rigorous justification for exchanging infinite sums and integrals, which would require additional assumptions, is not intended at this point.  It will be clear in Definition \ref{def:ub.N} that the $N$th-order approximation for $u$ contains only finite sums.}
Expression \eqref{eq:u.dyson} motivates the following definition:
\begin{definition}
\label{def:ub.N}
Let $u$ be given by \eqref{and1}.
Assume that for every $t \in [0,T]$ the coefficients $(a,b,c,f)$ of the
operator $\Ac(t)$ are $N$-times differentiable in the {spatial} variables $(x,y)$. For a fixed
{piecewise continuous map} $(\xb(\cdot),\yb(\cdot)): [0,T] \to \Rb^2$, the \emph{$N$th-order
approximation of $u$}, denoted $\ub_N$, is defined as
\begin{align}
\ub_N
    &=  \sum_{n=0}^N u_n , &
    &\text{where}&
u_0(t)
    &:= \Pc_0(t,T) \phi , \label{eq:u0.def}
\end{align}
and
\begin{align}
u_n(t)
    &:= \sum_{k=1}^n
            \int_{t}^T \dd t_1 \int_{t_1}^T \dd t_2 \cdots \int_{t_{k-1}}^T \dd t_k
            \\ & \qquad \sum_{i \in I_{n,k}}
            \Pc_0(t,t_1) \Ac_{i_1}(t_1)
            \Pc_0(t_1,t_2) \Ac_{i_2}(t_2) \cdots
            \Pc_0(t_{k-1},t_k) \Ac_{i_k}(t_k)
            \Pc_0(t_k,T) \phi \label{eq:un.def} .
\end{align}
Here, $\Ac_i(t)$ and $I_{n,k}$ are as given in \eqref{eq:A.expand} and \eqref{eq:Ink}, respectively, and $\Pc_0(t,T)$ is the semigroup generated by $\Ac_0(t)$.
\end{definition}

\subsection{Expression for $u_0$}
\label{sec:u0}

The action of the semigroup $\Pc_0(t,T)$ generated by $\Ac_0(t)$ when acting on a function
$\theta: \Rb^3 \to \Rb$ is
\begin{align}
\Pc_0(t,T) \theta(x,y,z)
    &=  \int_{\Rb^3} \dd \xi \dd \eta \dd \zeta \, \del_{\bar{z}} ( \zeta  ) \,
            \Gamma_0(t,x,y;T,\xi,\eta) \theta(\xi,\eta,\zeta), \label{eq:P0}
\end{align}
where $\del_{\bar{z}}$ is a Dirac mass centered at
\begin{align}
\bar{z}
    =z + \beta(\xi - x) - \beta(\beta-1)\int_0^t a_{0,0}(s) \dd s , \label{eq:z.bar}
\end{align}
and
\begin{align}
\Gamma_0(t,x,y;T,\xi,\eta)
    &=  \frac{1}{ 2\pi \sqrt{|\Cv|} } \exp\( -\frac{1}{2}\mv^\text{T} \Cv^{-1} \mv \) , \label{eq:Gamma.0}
\end{align}
with the covariance matrix $\Cv$ and vector $\mv$ given by:
\begin{align}
\Cv
    &=  \begin{pmatrix}
            2 \int_t^T a_{0,0}(s) \dd s & \int_t^T f_{0,0}(s) \dd s  \\
            \int_t^T f_{0,0}(s) \dd s & 2 \int_t^T b_{0,0}(s) \dd s
        \end{pmatrix} , &
\mv
    &=  \begin{pmatrix}
            \xi - x + \int_t^T a_{0,0}(s) \dd s \\ \eta - y - \int_t^T c_{0,0}(s) \dd s
         \end{pmatrix} .
\end{align}
Using \eqref{eq:u0.def}, we have $u_0(t)= \Pc_0(t,T) \phi$.  Hence, from \eqref{eq:P0} a direct
computation gives the zeroth-order approximation

\begin{align}
u_0(t,z)
    =   \int_\Rb \dd \zeta \,
            \frac{1}{\sqrt{2\pi s^2(t,T)}} \exp \( \frac{-(\zeta-m(t,T))^2}{2 s^2(t,T)} \)
            \phi(\zeta) ,   \label{eq:u0.bs}
\end{align}
where the mean $m(t,T)$ and variance $s^2(t,T)$ are given by
\begin{align}
m(t,T)
    &= z - \beta^2 \int_t^T \dd t_1 \, a_{0,0}(t_1), &
s^2(t,T)
    &=  2 \beta^2 \int_t^T \dd t_1 \, a_{0,0}(t_1).
\end{align}

\subsection{Expression for $u_n$}
\label{sec:un} The following theorem, and the ensuing proof, show that $u_n(t)$ can be written as
a differential operator $\Lc_n(t,T)$ acting on $u_0(t)$.  {The theorem is written specifically
for Put options, which play an important role in derivative markets.
Call prices, which are also important in derivative markets, can be obtained from Put prices via Put-Call parity.
}
\begin{theorem}
\label{thm:un}
Assume that for every $t \in [0,T]$ the coefficients $(a,b,c,f)$ of the operator $\Ac(t)$ are
$n$-times differentiable in the spatial variables $(x,y)$. {Assume also that $\varphi$ is the
payoff of a Put option on $Z$.  That is, $\varphi(z) = \left(\ee^k-\ee^z\right)^+$.} Then, for a
fixed {piecewise continuous map} $(\xb(\cdot),\yb(\cdot)): [0,T] \to \Rb^2$, the function $u_n$
defined in \eqref{eq:un.def} is given explicitly by
\begin{align}
u_n(t)
    &=  \Lc_n(t,T) u_0(t) , \label{eq:un} &
\end{align}
where $u_0$ is given by \eqref{eq:u0.bs} and
\begin{align}
\Lc_n(t,T)
    &=  \sum_{k=1}^n
            \int_{t}^T \dd t_1 \int_{t_1}^T \dd t_2 \cdots \int_{t_{k-1}}^T \dd t_k
            \sum_{i \in I_{n,k}}
            \Gc_{i_1}(t,t_1)
            \Gc_{i_2}(t,t_2) \cdots
            \Gc_{i_k}(t,t_k) , \label{eq:Ln}
\end{align}
with $I_{n,k}$ as defined in \eqref{eq:Ink} and
\begin{align}
\Gc_n(t,t_i)
    &:= \sum_{k=0}^n \(\Mc_x(t,t_i)-\xb(t_i)\)^{n-k} \(\Mc_y(t,t_i)-\yb(t_i)\)^{k}  \Ac_{n-k,k}(t_i)  \label{eq:G.def} \\
\Mc_x(t,t_i)
    &:= x + \int_{t}^{t_i} \dd s \Big( a_{0,0}(s) \( 2 \d_x + 2 \beta \d_z - 1 \)
                + f_{0,0}(s) \d_y \Big), \label{eq:Mx.def} \\
\Mc_y(t,t_i)
    &:= y + \int_{t}^{t_i} \dd s \Big( f_{0,0}(s) \( \d_x + \beta \d_z \)
                + 2 b_{0,0}(s) \d_y + c_{0,0}(s) \Big). \label{eq:My.def}
\end{align}
\end{theorem}
\begin{proof}
The proof consists of showing that the operator $\Gc_i(t,t_k)$ in \eqref{eq:G.def} satisfies
\begin{align}
\Pc_0(t,t_k) \Ac_{i}(t_k)
    &=  \Gc_i(t,t_k) \Pc_0(t,t_k) . \label{eq:PA=GP}
\end{align}
Assuming \eqref{eq:PA=GP} holds, we can use the fact that $\Pc_0(t,T)$ satisfies the semigroup property
\begin{align}
\Pc_0(t,T)
    &=  \Pc_0(t,t_1) \Pc_0(t_1,t_2) \cdots \Pc_0(t_{k-1},t_k) \Pc_0(t_k,T) , \label{eq:semigroup.prop}
\end{align}
and we can re-write \eqref{eq:un.def} as
\begin{align}
u_n(t)
    &=  \sum_{k=1}^n
            \int_{t}^T \dd t_1 \int_{t_1}^T \dd t_2 \cdots \int_{t_{k-1}}^T \dd t_k
            \\ & \qquad \sum_{i \in I_{n,k}}
            \Gc_{i_1}(t,t_1)
            \Gc_{i_2}(t,t_2) \cdots
            \Gc_{i_k}(t,t_k)
            \Pc_0(t,T) \phi . \label{eq:u.dyson2}
\end{align}
Note, in deriving \eqref{eq:u.dyson2}, we have repeatedly used \eqref{eq:PA=GP} to move the semigroup operators $\Pc_0(t_i,t_{i+1})$ in \eqref{eq:un.def} past the $\Ac_i(t_{i+1})$ operators.  Then, we used \eqref{eq:semigroup.prop}.
Finally, using that $\Pc_0(t,T) \phi=u_0(t)$, equations \eqref{eq:un}-\eqref{eq:Ln} follow directly from \eqref{eq:u.dyson2}  Thus, we only need to show that $\Gc_i(t,t_k)$ satisfies \eqref{eq:PA=GP}.
\par
To establish \eqref{eq:PA=GP}, we note that
\begin{align}
\Mc_x(t,T) \Big( \del_{\bar{z}}(\zeta) \Gamma_0(t,x,y;T,\xi,\eta) \Big)
    &=  \xi \Big( \del_{\bar{z}}(\zeta) \Gamma_0(t,x,y;T,\xi,\eta) \Big) , \label{eq:Mx.prop} \\
\Mc_y(t,T) \Big( \del_{\bar{z}}(\zeta) \Gamma_0(t,x,y;T,\xi,\eta) \Big)
    &=  \eta \Big( \del_{\bar{z}}(\zeta) \Gamma_0(t,x,y;T,\xi,\eta) \Big) , \label{eq:My.prop}
\end{align}
where $\bar{z}$ is defined in \eqref{eq:z.bar} and $\Gamma_0$ is defined in \eqref{eq:Gamma.0}.  This is a direct computation, which can be checked by hand.  It follows from repeated application of \eqref{eq:Mx.prop} and \eqref{eq:My.prop} that if $p: \Rb^2 \to \Rb $ is a polynomial function, we have
\begin{align}
p\(\Mc_x(t,T),\Mc_y(t,T)\)\Big( \del_{\bar{z}}(\zeta) \Gamma_0(t,x,y;T,\xi,\eta) \Big)
    &=  p(\xi,\eta) \Big( \del_{\bar{z}}(\zeta) \Gamma_0(t,x,y;T,\xi,\eta) \Big) . \label{p}
\end{align}
In what follows we write $\Ac_{n-k,k}^{\xi,\eta,\zeta}(s)$ and $\Ac_{n-k,k}^{x,y,z}(s)$ in order to indicate explicitly which variables these operators act on.  We also denote by $(\Ac_{n-k,k}^{\xi,\eta,\zeta}(s))^*$ the formal adjoint of $\Ac_{n-k,k}^{\xi,\eta,\zeta}(s)$.
Suppose $\theta:\Rb^3 \to \Rb$ is $C^2(\Rb^3)$ and at most exponentially growing.  Then we have
\begin{align}
&\Pc_0(t,s) \Ac_{i}(s) \theta(x,y,z) \\
    &=  \int_{\Rb^3} \dd \xi \dd \eta \dd \zeta \, \del_{\bar{z}}(\zeta) \Gamma_0(t,x,y;s,\xi,\eta)
            \sum_{k=0}^n  \(\xi-\xb(s)\)^{n-k} \(\eta-\yb(s)\)^{k} \Ac_{n-k,k}^{\xi,\eta,\zeta}(s) \theta(\xi,\eta,\zeta) \\
    &=  \sum_{k=0}^n \(\Mc_x(t,s)-\xb(s)\)^{n-k} \(\Mc_y(t,s)-\yb(s)\)^{k} \int_{\Rb^3} \dd \xi \dd \eta \dd \zeta \,
            \del_{\bar{z}}(\zeta) \Gamma_0(t,x,y;s,\xi,\eta)
            \Ac_{n-k,k}^{\xi,\eta,\zeta}(s) \theta(\xi,\eta,\zeta) \\
    &=  \sum_{k=0}^n \(\Mc_x(t,s)-\xb(s)\)^{n-k} \(\Mc_y(t,s)-\yb(s)\)^{k} \int_{\Rb^3} \dd \xi \dd \eta \dd \zeta \,
            \del_{\bar{z}}(\zeta) \theta(\xi,\eta,\zeta) \( \Ac_{n-k,k}^{\xi,\eta,\zeta}(s)\)^*\Gamma_0(t,x,y;s,\xi,\eta) \\
    &=  \sum_{k=0}^n \(\Mc_x(t,s)-\xb(s)\)^{n-k} \(\Mc_y(t,s)-\yb(s)\)^{k} \Ac_{n-k,k}^{x,y,z}(s) \int_{\Rb^3} \dd \xi \dd \eta \dd \zeta \,
            \del_{\bar{z}}(\zeta) \theta(\xi,\eta,\zeta) \Gamma_0(t,x,y;s,\xi,\eta) \\
    &=  \Gc_i(t,s) \Pc_{0}(t,s) \theta(x,y,z) .
\end{align}
The first equality follows from the definitions of $\Pc_0(t,s)$ and $\Ac_{i}(s)$.  In the second equality we have used \eqref{p} and pulled the operators $\Mc_x$ and $\Mc_y$ out of the integral since they act on the backward variables $(x,y,z)$.  In the third equality we have intergrated by parts.  In the fourth equality we have used the symmetry property of the kernel $\del_{\bar{z}}(z) \Gamma_0(t,x,y;s,\xi,\eta)$ to replace $\( \Ac_{n-k,k}^{\xi,\eta,\zeta}(s)\)^*$ with $\Ac_{n-k,k}^{x,y,z}(s)$.  We then pulled $\Ac_{n-k,k}^{x,y,z}(s)$ out of the integral since it acts on the backward variables $(x,y,z)$.  The last equality follows from the definitions of $\Gc_i(t,s)$ and $\Pc_0(t,s)$.  Thus, we have established $\Pc_0(t,s) \Ac_{i}(s)=\Gc_i(t,s) \Pc_{0}(t,s)$, when acting on a function $\theta$ that is $C^2(\Rb^3)$ that is at most exponentially growing.
\par
To complete the proof we must show that terms of the form
\begin{align}
\Pc_0(t,t_1) \Ac_{i_1}(t_1)  \Pc_0(t_1,t_2) \Ac_{i_2}(t_2) \cdots \Pc_0(t_{k-1},t_k) \Ac_{i_k}(t_k) \Pc_0(t_k,T) \phi , \label{eq:PAPAP}
\end{align}
are at least $C^2(\Rb^3)$ and at most exponentially growing.
In fact, we will show that such terms are $C_e^\infty(\Rb^3)$, where $C_e^\infty(\Rb^3)$ denotes the space of functions that are $C^\infty(\Rb^3)$ with derivatives of all orders that are at most exponentially growing.
To see this, we note that $\Pc_0(t_k,T) \phi = u^\BS(t_k)$, where $u^\BS$ is the Black-Scholes price of a put option.
Since derivatives of the Black-Scholes put price with respect to $z$ are $C_e^\infty(\Rb)$
it follows that $\Pc_0(t_k,T) \phi \in C_e^\infty(\Rb^3)$.  Now, note that $C_e^\infty(\Rb^3)$ is invariant under differentiation, multiplication by a polynomial, and transformation by the semigroup operator $\Pc_0(t,s)$.  It follows that any term of the form \eqref{eq:PAPAP} is a member of $C_e^\infty(\Rb^3)$.
\end{proof}
\begin{remark}
\label{rmk:call.payoff}
In fact, Theorem \ref{thm:un} holds directly for Call options as well, since derivatives of the Black-Scholes call price with respect to $z$ are $C_e^\infty(\Rb)$.
\end{remark}
In the following proposition, we provide an alternative characterization of the approximating sequence $(u_n)$ as the solution of a nested sequence of PDEs.  This alternative characterization, which was derived using alternative methods in \cite{lorig-pagliarani-pascucci-4}, will be used in Section \ref{sec:accuracy} for the analysis of the accuracy of the approximation.
\begin{proposition}
\label{thm:pde}
Let $\varphi$ be the payoff of a put option: $\varphi(z)=(\ee^k - \ee^z)^+$.
The sequence of functions $(u_n)$ in \eqref{eq:un}
solves the following nested sequence of Cauchy problems
\begin{align}
(\d_t + \Ac_0(t)) u_0
    &=  0 , &
u_0(T)
    &=  \varphi , \label{u0} \\
(\d_t + \Ac_0(t)) u_n
    &=  - \sum_{k=1}^n \Ac_k(t) u_{n-k}, &
u_n(T)
    &=  0 , &
n
    &\geq 1. \label{un}
\end{align}
\end{proposition}
\begin{proof}
The proof is by induction.  By Duhamel's principle, the solution to \eqref{u0} and the solution to \eqref{un} with $n=1$ are
\begin{align}
u_0(t)
    &= \Pc_0(t,T) \varphi , &
u_1(t)
    &=  \int_{t_1}^T \dd t_1 \Pc_0(t,t_1) \Ac_1(t_1) \Pc_0(t_1,T) \varphi ,
\end{align}
in agreement with \eqref{eq:u0.def} and \eqref{eq:un.def}.  We now assume expression \eqref{eq:un.def} holds for the first $(n-1)$ terms and show that it holds for the $n$th term.  Once again, using Duhamel's principle, the solution to \eqref{un} is
\begin{align}
u_n(t)
        &=  \sum_{k=1}^n \int_t^T \dd t_0 \Pc_0(t,t_0) \Ac_k(t_0) u_{n-k}(t_0) \\
        &=  \int_t^T \dd t_0 \Pc_0(t,t_0) \Ac_n(t_0) \Pc_0(t_0,T) \phi \\ & \qquad
                + \sum_{k=1}^{n-1} \int_t^T \dd t_0 \Pc_0(t,t_0) \Ac_k(t_0) \sum_{m=1}^{n-k}
                \int_{t_0}^T \dd t_1 \int_{t_1}^T \dd t_2 \cdots \int_{t_{m-1}}^T \dd t_{m} \\ & \qquad
                \sum_{i \in I_{n-k,m}}
            \Pc_0(t_0,t_1) \Ac_{i_1}(t_1)
            \Pc_0(t_1,t_2) \Ac_{i_2}(t_2) \cdots
            \Pc_0(t_{m-1},t_m) \Ac_{i_m}(t_m)
            \Pc_0(t_m,T) \phi \\
        &=  \int_t^T \dd t_0 \Pc_0(t,t_0) \Ac_n(t_0) \Pc_0(t_0,T) \phi \\ & \qquad
                + \sum_{k=1}^{n-1}
                \int_t^T \dd t_0 \int_{t_0}^T \dd t_1 \int_{t_1}^T \dd t_2 \cdots \int_{t_{m-1}}^T \dd t_{m} \\ & \qquad
                \sum_{m=1}^{n-k} \sum_{i \in I_{n-k,m}}
                        \Pc_0(t,t_0) \Ac_k(t_0)
            \Pc_0(t_0,t_1) \Ac_{i_1}(t_1)
            \Pc_0(t_1,t_2) \Ac_{i_2}(t_2) \cdots
            \Pc_0(t_{m-1},t_m) \Ac_{i_m}(t_m)
            \Pc_0(t_m,T) \phi \\
        &=  \sum_{k=1}^n
            \int_{t}^T \dd t_1 \int_{t_1}^T \dd t_2 \cdots \int_{t_{k-1}}^T \dd t_k
            \\ & \qquad \sum_{i \in I_{n,k}}
            \Pc_0(t,t_1) \Ac_{i_1}(t_1)
            \Pc_0(t_1,t_2) \Ac_{i_2}(t_2) \cdots
            \Pc_0(t_{k-1},t_k) \Ac_{i_k}(t_k)
            \Pc_0(t_k,T) \phi ,
\end{align}
which agrees with expression \eqref{eq:un.def}.
\end{proof}
\begin{remark}
\label{rmk:bs}
Note that, by \eqref{eq:u0.bs}, the order zero price $u_0$ is simply an integral of the option payoff $\phi$ versus a Gaussian kernel $\Gam_0$, just as in the Black-Scholes model.  From Theorem \ref{thm:un} we see that higher order terms $u_n$ can be obtained by applying the differential operator $\Lc_n$ to $u_0$.  The operator $\Lc_n$ acts on the backward variable $z$, which is present only in the Gaussian kernel $\Gam_0$, producing (Hermite) polynomials in the forward variable $\zeta$ multiplied by $\Gam_0$.  Thus, every term in the price expansion is of the form
\begin{align}
u_n(t,z)
    =   \int_\Rb \dd \zeta \,
            \frac{p_n(\zeta)}{\sqrt{2\pi s^2(t,T)}} \exp \( \frac{-(\zeta-m(t,T))^2}{2 s^2(t,T)} \)
            \phi(\zeta) .
\end{align}
where the function $p_n$ is a polynomial.  As such, computation times for approximate prices are comparable to the Black-Scholes model.
\end{remark}

%
%

\section{Accuracy of the option-pricing approximation}
\label{sec:accuracy} The goal of this section is to establish a rigorous error bound for the
$N$th-order pricing approximation described in the previous sections.  We will adapt the
methods
from \cite{PaglPas2014}, 
who treat operators $\Ac(t)$ that are {\it locally elliptic}, to our current case, where the
operator $\Ac(t)$ is singular (see Remark \ref{rmk:elliptic}).  Our main error bound is given in
{Theorem} \ref{cor:u} at the end of this section.  In order to prove this {theorem} we introduce
$A(t,x,y)$, the  symmetric and positive semi-definite diffusion matrix of the $(X,Y)$ process:
\begin{align}
A(t,x,y)
    &:= \frac{1}{2}
            \begin{pmatrix}
            2 a(t,x,y) & f(t,x,y) \\ f(t,x,y) & 2 b(t,x,y)
            \end{pmatrix} .
\end{align}
We also introduce $D_{r}(x_{0},y_{0})$, the Euclidean ball
\begin{align}
D_{r}(x_{0},y_{0})=\{(x,y)\in\R^{2} : \left|(x,y)-(x_{0},y_{0})\right|<r\} ,
\end{align}
which is defined for any $(x_{0},y_{0}) \in \R^{2}$ and $r>0$.
Throughout Section \ref{sec:accuracy} we assume the following:
\begin{assumption}
\label{assumption2}
The function $u$ in \eqref{and1} solves the backward Cauchy problem
\begin{align}
 \( \d_t + \Ac(t) \) u(t,x,y,z)
    &=  0, &
(t,x,y,z)
    &\in[0,T) \times D , \label{eq:u.pde2} \\
u(T,x,y,z)
    &=  \phi(z), &(x,y,z)
    &\in D,
\end{align}
where $D$ is a domain in $\R^{3}$.  It is possible, but not required, that $D=\R^{3}$.
\end{assumption}
\begin{assumption}
\label{assumption1}
\begin{enumerate}
\item[]
\item[i)] 
{\it Local boundedness and global regularity}: the coefficients $a,b,c,f$ belong to $L^{\infty}_{\text{\rm
loc}}([0,T]\times D)$ and satisfy
$a(t,\cdot,\cdot),b(t,\cdot,\cdot),c(t,\cdot,\cdot),f(t,\cdot,\cdot)\in C^{N+1}(D)$ for any
$t\in[0,T]$.
\item[ii)] {\it {Local} non-degeneracy}: the diffusion matrix of $A=A(t,x,y)$ is positive definite on some cylinder $[0,T]\times D_{r}(x_{0},y_{0})$.
More precisely, $A=\tilde{A}$ in $[0,T]\times D_{r}(x_{0},y_{0})$ where $\tilde{A}\in
L^{\infty}([0,T]\times \R^{2})$ is a matrix of the form
\begin{align}
 \tilde{A}(t,x,y)= \frac{1}{2}
            \begin{pmatrix}
            2 \at(t,x,y) & \ft(t,x,y) \\ \ft(t,x,y) & 2 \bt(t,x,y)
            \end{pmatrix} , \label{amat}
\end{align}
such that
$\tilde{A}(t,\cdot,\cdot) \in
C_{\text{b}}^{N+1}(\R^{2})$ for any $t\in[0,T]$,
where $C_{\text{b}}^{N}$ denotes the space of continuously differentiable functions
with bounded derivatives up to order $N$,
and
\begin{align}\label{cond:parabolicity}
 M^{-1}|\xi|^2\leq \sum_{i,j=1}^{2}\tilde{A}_{ij}(t,x,y)\xi_{i}\xi_{j}\leq M |\xi|^2, \qquad t\in[0,T],\  (x,y),\xi\in\mathds{R}^2,
\end{align}
for some positive constant $M$. We also require the existence of a function $\ct\in
L^{\infty}([0,T]\times \R^{2})$ such that $\ct(t,\cdot,\cdot)\in C_{\text{b}}^{N+1}(\R^{2})$ for
any $t\in[0,T]$ and $c=\ct$ in $[0,T]\times D_{r}(x_{0},y_{0})$.
\end{enumerate}
\end{assumption}
\begin{assumption}
\label{assumption3}
We assume the payoff function $\phi$ is that of a Put option on $Z$.
That is, $\phi(z)=\left(\ee^k-\ee^z\right)^+$.
\end{assumption}
\begin{remark}
\label{rmk:assumption}
Assumptions \ref{assumption2} and \ref{assumption1} are satisfied by a number of well-known models
including Heston, constant elasticity of variance (CEV) and SABR. Thus, for these models,
we can establish rigorous error bounds for European Put prices.
Error bounds for Call prices, obtained via Put-Call parity, retain the same order of accuracy.
\end{remark}
\begin{remark}
\label{rmk:matt} Note that, although the diffusion matrix $A$ of the process $(X,Y)$ is locally
positive definite, the diffusion matrix \eqref{eq:matrix} of the process $(X,Y,Z)$ remains
singular.  Thus, Cauchy problem \eqref{eq:u.pde} is \emph{not parabolic} at any point.  This
issue, which is not handled in \cite{PaglPas2014}, 
presents a technical challenge that must be overcome in order to establish error estimates for our
pricing approximation $\ub_n$.
\end{remark}
\par
In order to cope with the double degeneracy of the pricing operator (recall, we have a partial
degeneracy in the $(x,y)$ variables and a global degeneracy in the $z$ variable), we now use an
elliptic regularization technique.  Specifically, we introduce a process $Z^\eps$, which is a
modification of the dynamics of $Z$ in \eqref{eq:SDE.0}.  We define
\begin{align}
\dd Z_t^\eps
    &:= \dd Z_t - \frac{1}{2} \eps^2 \dd t + \eps \dd W_t^z , &
\dd \< W^x, W^z\>
    &=  0 , &
\dd \< W^y, W^z\>
    &=  0 , &
\eps
    &\geq 0 .
\end{align}
We denote by $\Ac^{\e}(t)$ the 
infinitesimal generator of the Markov process $(X,Y,Z^\eps)$ and by $u^\eps$ the solution of
Cauchy problem related to $\Ac^{\e}(t)$, with final datum $\phi$. Specifically
\begin{align}
 \( \d_t + \Ac^\eps(t) \) u^\eps(t,x,y,z)
    &=  0, &
(t,x,y,z)
    &\in[0,T) \times D_{\e} , \label{eq:u.pde2eps} \\
u^\eps(T,x,y,z)
    &=  \phi(z), &(x,y,z)
    &\in D_{\e},
\end{align}
where $D_{\e}$ is some domain  of $\R^{3}$.
Thus, $u^\eps$ represents the price of a European {Put} option written on $(X,Y,Z^\eps)$. 
{Assumption \ref{assumption1}-i)} guarantees we can construct $\ub_N^\eps$, {the $N$-th
order approximation of $u^\eps$}, by replacing $\Ac(t)$ with $\Ac^\eps(t)$ in Definition
\ref{def:ub.N}.  Moreover, {for any $\e>0$,} $\Ac^{\e}(t)$ and $\varphi$ satisfy the assumptions of
Theorem 3.1 in \cite{PaglPas2014}, in which local error estimates for $|{u}^\eps - \ub_N^\eps|$
are established.
Below we prove that such error estimates are {\it uniform in $\e$} and therefore
error bounds for the price approximation $\bar{u}_N$ of options written on the Markov process
$(X,Y,Z)$ will follow in the limit as $\eps \to 0$.
\par
It is useful at this point to introduce the process $V^\eps$, which satisfies the following SDE:
\begin{align}
\dd V_t^\eps
    &=  \frac{1}{2}\(\beta(1-\beta) \sig^2(t,X_t,Y_t) - \eps^{2} \)\dd t + \eps \dd W_t^z .
\end{align}
We note that the dynamics of $Z^\eps$ can be written as follows
\begin{align}\label{e10}
\dd Z_t^\eps
    &=  \beta \, \dd X_t + \dd V_t^\eps .
\end{align}
Therefore, rather than considering the generator of $(X,Y,Z^\eps)$, we can consider the generator
of $(X,Y,V^\eps)$, which (with a slight abuse of notation) we denote again by $\Ac^\eps(t)$.  This
operator separates into an operator $\Xc(t)$, which takes derivatives with respect to $(x,y)$, and
an operator $\Vc^\eps(t)$, which takes derivatives with respect to $v$.  That is,
\begin{align}
 \Ac^\eps(t) &=  \Xc(t) + \Vc^\eps(t), \label{ae2} \\
\Xc(t)
    &=  a(t,x,y) \( \d_x^2 - \d_x \) + b(t,x,y) \d_y^2+ f(t,x,y) \d_x \d_y+c(t,x,y) \d_y, \\
\Vc^\eps(t)
    &=  \frac{\e^{2}}{2} \( \d_v^2 - \d_v \)+a(t,x,y)\b(1-\b) \d_v .
\end{align}

The first step in the proof of Theorem 3.1 in \cite{PaglPas2014} consists of extending operator
$\Ac^{\e}(t)$, which is defined on $[0,T]\times D$, to a uniformly elliptic operator
$\Act^{\e}(t)$ on $[0,T]\times \R^{3}$. This can be done by virtue of Assumption
\ref{assumption1}-ii). Indeed, for any $\e\ge 0$, it suffices to define {
\begin{align}
 \Act^\eps(t) &=  \Xct(t) + \Vct^\eps(t), \\
\Xct(t)
    &=  \at(t,x,y) \( \d_x^2 - \d_x \) + \bt(t,x,y) \d_y^2+ \ft(t,x,y) \d_x \d_y+ {\ct(t,x,y)} \d_y, \\
\Vct^\eps(t)
    &=  \frac{\e^{2}}{2} \( \d_v^2 - \d_v \) + \at (t,x,y)\b(1-\b) \d_v .
\end{align}
By Assumption \ref{assumption1}, $\Ac^{\e}(t)=\Act^{\e}(t)$ and $\Xc(t)=\tilde{\Xc}(t)$ in $[0,T] \times D_{r}(x_{0},y_{0})\times \R$.}
Notice that $\Act^{\e}(t)$ and $\tilde{\Xc}(t)$ are uniformly elliptic
operators on $[0,T]\times \R^{3}$ and $[0,T]\times \R^{2}$ respectively. Moreover,
$(\d_t+\Act^\eps(t))$ is uniformly parabolic and has a fundamental solution, denoted by
\begin{align}\label{ae1}
 \Gt^{\e}=\Gt^\eps(t,x,y,v;T,x',y',v'),\qquad t<T ,
\end{align}
{which (by definition) is the solution to
\begin{align}
 (\d_t + \Act^\eps(t) ) \Gt^\eps(t,x,y,v;T,x',y',v')
    &=  0 , & (t,x,y,v)
    &\in [0,T) \times \Rb^3 , \\
 \Gt^\eps(T,\cdot,\cdot,\cdot;T,x',y',v') &= { \del_{x',y',v'}.}
\end{align}
}
In the following lemma, we show that $\Gt^{\e}$ satisfies some Gaussian estimates.
\begin{lemma}\label{t1}
Let $i,j,h,k\in\N_{0}$ with $h+k\le N+2$, and $\bar{T}>0$.
{Then, under Assumption \ref{assumption1}, we have}
\begin{align}\label{e1}
  \left|(x-x')^{i}(y-y')^{j}\p_{x}^{h}\p_{y}^{k}\Gt^\eps(t,x,y,v;T,x',y',v')\right|\le \mathbf{c}_{0}(T-t)^{\frac{i+j-h-k}{2}}
  \G^{(M,\e)}_{\text{\rm heat}}(t,x,y,v;T,x',y',v')
\end{align}
for any $x,y,v,x',y',v'\in\R$, $0\le t<T\le \bar{T}$ and $\e \in (0,1]$.  Here,
$\G^{(M,\e)}_{\text{\rm heat}}$ denotes the fundamental solution of the heat operator
\begin{align}\label{eq:heatoperator}
\d_t + M (\p_{xx}+\p_{yy})+\frac{\e^{2}}{2}\p_{vv},
\end{align}
and $\mathbf{c}_{0}$ is a positive constant that depends only on $M,N,i,j$ and $\bar{T}$. In
particular, the constant $\mathbf{c}_{0}$ is independent of $\eps$.
\end{lemma}
\begin{proof}
Estimate \eqref{e1}  differs slightly from the classical Gaussian estimates for parabolic equations
(cf. \cite{friedman-parabolic}; see also \cite{DiFrancescoPascucci2}, \cite{pascuccibook} for a
more recent and general presentation) because the operator $(\d_t + \Act^\eps(t))$, while
parabolic, is \emph{not uniformly parabolic} with respect to $\eps\in\,(0,1]$. Nevertheless, the
thesis can be proved by mimicking the classical argument which is based on the parametrix method
and carefully checking that the constant $\mathbf{c}_0$ is independent of $\eps$.  In particular,
the main ingredients in the parametrix construction are some uniform-in-$\eps$, Gaussian
estimates (see, for instance, Proposition 3.1 in \cite{DiFrancescoPascucci2}), which we now describe.
For any fixed $(\bar{x},\bar{y})\in\R^{2}$, we denote by $\tilde{\Xc}_{\xb,\yb}(t)$ the
operator obtained by freezing at $(\bar{x},\bar{y})$ the coefficients of $\tilde{\Xc}(t)$ and we set
\begin{align}
 \Act^\eps_{\bar{x},\bar{y}}(t)
    &:= \Xct_{\bar{x},\bar{y}}(t)+ \frac{\e^{2}}{2}\d_v^2. 
\label{eq:Afrozen}
\end{align}
Let $\Gt^\eps_{\bar{x},\bar{y}}$ and $\Gt_{\bar{x},\bar{y}}$ be the fundamental solutions
corresponding to $(\d_t+\Act_{\xb,\yb}^\eps)$ and $(\d_t+\Xct_{\xb,\yb})$ respectively. Then for
every $\bar{x},\bar{y},x,y,v,x',y',v'\in\R$, $0\le t<T\le \bar{T}$ and $\eps\in\,(0,1]$, we have
\begin{align}\label{e2}
  M^{-2}\G^{(M^{-1},\e)}_{\text{\rm heat}}(t,x,y,v;T,x',y',v')\le
  \Gt^\eps_{\bar{x},\bar{y}}(t,x,y,v;T,x',y',v')\le M^{2}\G^{(M,\e)}_{\text{\rm heat}}(t,x,y,v;T,x',y',v').
\end{align}
Estimate \eqref{e2} can be readily proved as in Proposition 3.1 in \cite{DiFrancescoPascucci2}, by
noting that $\Gt^\eps_{\bar{x},\bar{y}}=\Gt_{\bar{x},\bar{y}}\G_{\e}$ where $\G_{\e}$ is the
fundamental solution of the one-dimensional heat (parabolic) operator
$(\d_t+\frac{\e^{2}}{2}\p_{vv})$. Notice that \eqref{e2} is uniform in $\eps$ (i.e. the constants
in the estimates are independent of $\eps$).  Based on this fact, the estimate \eqref{e1}, with
$\mathbf{c}_{0}$ independent of $\eps$, follows by the parametrix method.
\end{proof}

\begin{lemma}
\label{th:error_estimates_taylor} Let Assumption \ref{assumption1} hold. {Denote by $\Gt^\eps$ the
fundamental solution in \eqref{ae1} corresponding to $(\d_t + \Act^\eps(t))$. Denote by
$\bar{\Gam}_N^\eps$ the $N$th-order approximation of $\Gt^\eps$, constructed using
$(\xb(\cdot),\yb(\cdot))=(x,y)$.} Then we have
\begin{align}\label{th:error_estim_fund_solution}
 \left|\Gt^\eps(t,x,y,v;T,x',y',v')-\bar{\Gamma}^\eps_{N}(t,x,y,v;T,x',y',v')\right|
 \leq \mathbf{c}_{1} (T-t)^{\frac{N+1}{2}}\G^{(M,\e)}_{\text{\rm heat}}(t,x,y,v;T,x',y',v'),
\end{align}
for any $x,y,v,x',y',v'\in\mathds{R}$, $0\leq t<T$ and $\eps\in\,(0,1]$, where $\mathbf{c}_1$ is a
positive constant that depends on $M,N,T$ but is independent of $\eps$.
\end{lemma}
\begin{proof}
Using the uniform in $\eps$ estimate \eqref{e1} and the ellipticity of  $\Ac^\eps(t)$, we can
repeat step by step the proof of \cite[Theorem 3.10]{lorig-pagliarani-pascucci-4}. The key ingredient in the
modified proof is to verify
that, since $\mathbf{c}_{0}$ in \eqref{e1} does not depend on $\eps$, neither does
$\mathbf{c}_{1}$.
\end{proof}

We are now in a position to state our main error estimate. For any $\eps \geq 0$, let
$\tilde{u}^\eps$ be the classical bounded 
solution of Cauchy problem
\begin{align}
 \( \d_t + \Act^{\e}(t) \) \tilde{u}^\eps(t,x,y,z)
    &=  0, &
(t,x,y,z)
    &\in[0,T)\times \R^{3},\\
 \tilde{u}^\eps(T,x,y,z)
    &= \left(\ee^k-\ee^z\right)^+,  &
(x,y,z)
    &\in \R^{3}.
\end{align}
For $\e=0$ we will generally omit the superscript and simply write $\tilde{u}$ instead of
$\tilde{u}^0$. 
\begin{theorem}
\label{cor:u} Let Assumptions \ref{assumption2},  \ref{assumption1} and \ref{assumption3} hold.
Let $\bar{u}^{\e}_N$, $\e\ge 0$, denote the $N$-th order approximation of $\tilde{u}^{\e}$, which is constructed as in \eqref{eq:u0.def} with $(\xb(\cdot),\yb(\cdot))=(x,y)$ and with $\Ac(t)$
replaced by $ \Act^{\e}(t)$.
Let $\bar{u}_N := \bar{u}_N^\eps |_{\eps=0}$.
{Note that $\bar{u}_N$ coincides with the $N$-th order approximation of $u$ when $(x,y,z)\in \left(D_{\delta r}(x_{0},y_{0})\times\R\right)\cap D.$}
Then for any
$\delta\in\,(0,1)$ we have
\begin{align}
 \left| u(t,x,y,z)-\bar{u}_N(t,x,y,z) \right| \leq \mathbf{c}_{2}(T-t)^{\frac{N+2}{2}}, \qquad 0\leq t<T,\ (x,y,z)\in \left(D_{\delta r}(x_{0},y_{0})\times\R\right)\cap D.
\end{align}
The constant $\mathbf{c}_{2}$ depends only on $\delta,k,M,N$ and $T$.
\end{theorem}
\begin{proof} Firstly, we remark explicitly that if $(x,y)\in D_{r}(x_{0},y_{0})$ then,
for any $\e\ge 0$, $\bar{u}^{\e}_N$ coincides with the $N$-th order approximation of $u^{\e}$
because $\Ac(t)\equiv \Act(t)$ in $[0,T]\times D_{r}(x_{0},y_{0})\times \R$. Then, integrating
estimate \eqref{th:error_estim_fund_solution} against the payoff function we obtain
\begin{align}\label{eq:error_estimate3}
 \left|\tilde{u}^{\e}(t,x,y,z)-\bar{u}^{\e}_N(t,x,y,z) \right| \leq \mathbf{c}_{1}(T-t)^{\frac{N+1}{2}}, \qquad 0\leq t<T,\ (x,y,z)\in \R^{3}.
\end{align}
By exploiting the Lipschitz regularity and boundedness of the Put payoff, we have
a more refined estimate with the power $\frac{N+2}{2}$ replacing $\frac{N+1}{2}$ in the exponent
of $(T-t)$ in \eqref{eq:error_estimate3}.
{Since the operator $\Ac^\eps(t)$, $\e>0$, and the payoff function $\varphi$ satisfy the
assumptions of Theorem 3.1 in \cite{PaglPas2014}, we can pass from the global error
estimate for $\tilde{u}^{\e}$ to the local estimate for $u^\eps$}
\begin{align}\label{eq:error_estimate2}
 \left| u^{\e}(t,x,y,z)-\bar{u}^{\e}_N(t,x,y,z) \right| \leq \mathbf{c}_{2} (T-t)^{\frac{N+2}{2}}, \qquad 0\leq t<T,\ (x,y,z)\in \left(D_{\delta
 r}(x_{0},y_{0})\times\R\right)\cap D.
\end{align}
By Lemma \ref{t1}, the above estimate is uniform in $\e$ and this is sufficient to conclude the
proof.
\end{proof}

%
%

\section{Implied volatility}
\label{sec:imp.vol} In this section, we translate our price expansion for a call option with
payoff function $\phi(z) = (\ee^z - \ee^k)^+$ into an expansion in implied volatility.  To ease
notation we shall suppress much of the dependence on $(t,T,x,y,z,k)$.  However, one should keep in
mind that prices and implied volatilities do depend on these quantities, even if this is not
explicitly indicated.  We begin our analysis by recalling the definitions of the Black-Scholes
call price and implied volatility.
\begin{definition}
\label{def:uBS}
The \emph{Black-Scholes Call price} $u^{\BS}: \Rb^+ \to \Rb^+$ is given by
\begin{align}
u^{\BS}(\sig)
    &:=     \ee^z \Nc(d_{+}(\sig)) - \ee^k \Nc(d_{-}(\sig)) , &
d_{\pm}(\sig)
    &:= \frac{1}{\sig \sqrt{\tau}} \( z - k \pm \frac{\sig^2 \tau}{2}  \) , &
\tau
        &:= T-t , \label{eq:uBS}
\end{align}
where $\Nc$ is the CDF of a standard normal random variable.
\end{definition}
\begin{definition}
\label{def:imp.vol.def} For fixed $(t,T,z,k)$, the \emph{implied volatility} corresponding to a call price $u\in\,((\ee^{z}-\ee^{k})^+,\ee^z)$ is defined as the unique strictly positive real solution $\sig$ of the equation
\begin{align}
 u^{\BS}(\sig)    &= u ,   \label{eq:imp.vol.def}
\end{align}
where $u^\BS$ is given by \eqref{eq:uBS}.
\end{definition}
\begin{theorem}
\label{thm:uBS}
For a European call option with payoff function $\phi(z)=(\ee^z-\ee^k)^+$ we have
\begin{align} \label{eq:zeroIV}
u_0
    &=  u^{\BS}(\sig_0) , &
\sig_0^2
    &=  \frac{2 \beta^2}{T-t} \int_t^T \dd s \, a_{0,0}(s) .
\end{align}
\end{theorem}
\begin{proof}
The proof follows directly from \eqref{eq:u0.bs} with $\phi(z)=(\ee^z-\ee^k)^+$.
\end{proof}
\noindent
From Theorem \ref{thm:uBS} we note that the price expansion \eqref{eq:u0.def} is of the form
\begin{align}
u
    &=  u^\BS(\sig_0) + \sum_{n=1}^\infty u_n . \label{eq:u.form}
\end{align}
As shown in \cite{lorig-pagliarani-pascucci-2} and \cite{lorig-jacquier-1}, the special form \eqref{eq:u.form} lends itself to an expansion
\begin{align}
\sig
    &=  \sig_0 + \eta , &
\eta
    &=  \sum_{n=1}^\infty \sig_n , \label{eq:sigma.expand}
\end{align}
of implied volatility.  To see this, one expands $u^\BS(\sig)$ as a Taylor series about the point $\sig_0$.
{For $\eta$ small enough (i.e., within the radius of convergence of the Taylor series expansion of $u^\BS$ about the point $\sig_0$) we have}
\begin{align}
u^\BS(\sig)
    &=  u^\BS(\sig_0 + \eta) \\
    &=  u^\BS(\sig_0) + \eta    \,\d_\sig u^\BS(\sig_0) + \frac{1}{2!} \eta^2   \d_\sig^2 u^\BS(\sig_0)
            + \frac{1}{3!} \eta^3   \d_\sig^3 u^\BS(\sig_0) + \ldots . \label{eq:uBS.expand}
\end{align}
Inserting expansions \eqref{eq:u.form} and \eqref{eq:uBS.expand} into equation \eqref{eq:imp.vol.def}, one can solve iteratively for every term in the sequence $(\sig_n)_{n \geq 1}$.
We define the \emph{$n$th-order approximation of implied volatility} as
\begin{align}
\sigb_n
    &=  \sum_{k=0}^n \sig_n .
\end{align}
The first four terms in the sum, which are enough to provide an accurate approximation of implied volatility, are $\sig_0$, given by \eqref{eq:zeroIV}, and
\begin{align}
\sig_1
    &=  \frac{u_1}{\d_\sig u^\BS(\sig_0)}, &
\sig_2
    &=  \frac{u_2 - \tfrac{1}{2} \sig_1^2 \d_\sig^2 u^\BS(\sig_0)}{\d_\sig u^\BS(\sig_0)} ,
    &
\sig_3
    &=  \frac{u_3 - \(\sig_2 \sig_1 \d_\sig^2 + \tfrac{1}{3!}\sig_1^3 \d_\sig^3 \) u^\BS(\sig_0)}{\d_\sig u^\BS(\sig_0)} . \label{eq:sig.n}
\end{align}
A general expression for the $n$th-order term can be found in \cite{lorig-pagliarani-pascucci-2,lorig-jacquier-1}.
\par
As written, the expressions in \eqref{eq:sig.n} are not particularly useful.  Indeed $u^\BS(\sig_0)$ and $u_n$ are Gaussian integrals, which are not numerically intensive to compute, but do not give much explicit information about how implied volatility depends on $(t,T,x,y,z,k,\beta)$.  However, using \eqref{eq:uBS} a direct computation shows
\begin{align}
\frac{\d_\sig^2u^\BS(\sig)}{\d_\sig u^\BS(\sig)}
    &=  \frac{(k-z)^2}{\tau \sig^3}-\frac{\tau \sigma }{4} , &
\frac{\d_\sig^3u^\BS(\sig)}{\d_\sig u^\BS(\sig)}
    &=  \frac{(k-z)^4}{\tau^2 \sig^6}-\( \frac{3}{\tau \sig^4}-\frac{1}{2 \sig^2}\)(k-z)^2
            +\frac{\tau^2 \sig^2}{16} -\frac{\tau}{4} . \label{eq:d.sigma}
\end{align}
In general, every term of the form $\d_\sig^n u^\BS(\sig_0)/\d_\sig u^\BS(\sig_0)$ can be computed explicitly.  Moreover, terms of the form $u_n/\d_\sig u^\BS(\sig_0)$ can also be computed explicitly.
To see this, we note from Theorems \ref{thm:un} and \ref{thm:uBS} that
\begin{align}
u_n
    &=  \Lc_n(t,T) u_0
    =   \widetilde{\Lc}_n(t,T)u^\BS(\sig_0) ,  \label{eq:clear}
\end{align}
where
\begin{align}
\widetilde{\Lc}_n(t,T)
    &=  \sum_{k=1}^n
            \int_{t}^T \dd t_1 \int_{t_1}^T \dd t_2 \cdots \int_{t_{k-1}}^T \dd t_k
            \sum_{i \in I_{n,k}}
            \Gc_{i_1}(t,t_1) \cdots
            \Gc_{i_{k-1}}(t,t_{k-1})
                        \widetilde{\Gc}_{i_k}(t,t_k) ,  \label{eq:L.tilde} \\
\widetilde{\Gc}_n(t,t_i)
        &:= \sum_{k=0}^n \(\Mc_x(t,t_i)-\xb(t_i)\)^{n-k} \(\Mc_y(t,t_i)-\yb(t_i)\)^{k}
                a_{n-k,k}(t_i) \beta^2  (\d_z^2-\d_z) u^\BS(\sig_0) .
\end{align}
Thus, $u_n$ is a finite sum of the form
\begin{align}
u_n
    &=  \sum_m \Chi_{n,m} \d_z^m (\d_z^2 - \d_z) u^\BS(\sig_0) , \label{eq:un.sum}
\end{align}
where the coefficients $(\Chi_{n,m})$ are $(t,T,x,y)$-dependent constants, which can be computed from Theorem \ref{thm:un}.  Now, using \eqref{eq:uBS}, a direct computation shows
\begin{align}
\frac{\d_z^m (\d_z^2 - \d_z ) u^{\BS}(\sig_0)}{\d_\sig u^{\BS}(\sig_0)}
    &= \(\frac{-1}{\sqrt{2 \sig_0^2 \tau}}\)^{m} \frac{H_{n}(w)}{\tau \sig_0} , &
w
    &:= \frac{z-k-\frac{1}{2}\sig_0^2 \tau}{\sig\sqrt{2\sig_0^2 \tau}} , \label{eq:eq}
\end{align}
where $H_n(z) :=  (-1)^n \ee^{z^2} \d_z^n \ee^{-z^2}$ is the $n$-th Hermite polynomial.
Combining \eqref{eq:un.sum} with \eqref{eq:eq} we have
\begin{align}
\frac{u_n}{\d_\sig u^\BS(\sig_0)}
    &=  \sum_m \Chi_{n,m} \(\frac{-1}{\sqrt{2 \sig_0^2 \tau}}\)^{m} \frac{H_{n}(w)}{\tau \sig_0} . \label{eq:Un}
\end{align}
Finally, from \eqref{eq:d.sigma} and \eqref{eq:Un}, we see that all terms in the implied volatility expansion \eqref{eq:sig.n} are polynomials in \emph{$\log$-moneyness} $\lam:=(k-z)$.  Explicit expressions for $(\sig_n)_{n \leq 3}$ under different models will be given in Section \ref{sec:examples}.   A general expression for $(\sig_n)_{n \leq 2}$ in the time-homogeneous LSV setting is given below.  We denote by
\begin{align}
\lam
    &=  k-z , &
\tau
    &=  T-t , &
(X_t,Y_t)
    &=  (x,y) ,
\end{align}
and we choose the expansion point of our Taylor series approximation as $(\xb(\cdot),\yb(\cdot))=(x,y)$.  We have
\begin{small}
\begin{align}
\label{eq:sig0a00}
\sig_0
    &= |\beta|\sqrt{ 2 a_{0,0} } , &
\sig_1
    &=  \sig_{1,0} + \sig_{0,1} , &
\sig_2
    &=  \sig_{2,0} + \sig_{1,1} + \sig_{0,2} ,
\end{align}
where
\begin{align}
\sig_{1,0}
   &=   \frac{\tau}{4} \( (\beta -1) \sigma _0 a_{1,0} \)
                + \frac{1}{2 \sigma _0}\( {\beta  a_{1,0}} \) \lam , \\
\sig_{0,1}
   &=       \frac{\tau}{4 \sigma _0} \( {\beta ^2 a_{0,1} \left(2 c_{0,0}+\beta  f_{0,0}\right)} \)
                + \frac{1}{2 \sigma _0^3} \( {\beta ^3 a_{0,1} f_{0,0}} \) \lam, \\
\sig_{2,0}
   &=   \frac{\tau}{24 \sig_0} \( 2 \sigma_0^2 a_{2,0} - 3 \beta ^2 a_{1,0}^2 \)
                + \frac{\tau^2}{96 \beta ^2} \( {\beta ^2 (2 \beta  (2 \beta -5)+5) \sigma _0 a_{1,0}^2+4 (\beta -1)^2 \sigma _0^3 a_{2,0}} \) \\ & \qquad
                + \frac{\tau}{24 \beta  \sigma _0} \( -{(\beta -1) \left(\beta ^2 a_{1,0}^2-4 \sigma _0^2 a_{2,0}\right)} \) \lam
                +   \frac{1}{12 \sigma _0^3} \( {2 \sigma _0^2 a_{2,0}-3 \beta ^2 a_{1,0}^2} \) \lam^2 , \\
\sig_{1,1}
   &=   \frac{\tau}{12 \sigma _0^3} \( {\beta ^2 \left(a_{0,1} \left(\beta ^2 a_{1,0} f_{0,0}-2 \sigma _0^2 f_{1,0}\right)+\sigma _0^2 a_{1,1} f_{0,0}\right)} \) \\ & \qquad
                + \frac{\tau^2}{48 \sigma _0} \( {a_{0,1} \left(\beta ^2 a_{1,0} \left(2 (\beta -1) c_{0,0}-\beta  f_{0,0}\right)+2 (\beta -1) \sigma _0^2 \left(2 c_{1,0}+\beta  f_{1,0}\right)\right)+2 (\beta -1) \sigma _0^2 a_{1,1} \left(2 c_{0,0}+\beta  f_{0,0}\right)} \) \\ & \qquad
                + \frac{\tau}{24 \sigma _0^3} \( {\beta  \left(a_{0,1} \left(5 \beta ^2 a_{1,0} \left((1-2 \beta ) f_{0,0}-2 c_{0,0}\right)+2 \sigma _0^2 \left(2 c_{1,0}+(2 \beta -1) f_{1,0}\right)\right)+2 \sigma _0^2 a_{1,1} \left(2 c_{0,0}+(2 \beta -1) f_{0,0}\right)\right)} \) \lam \\[-1.5em] & \qquad
                +   \frac{1}{6 \sigma _0^5} \( {\beta ^2 \left(a_{0,1} \left(\sigma _0^2 f_{1,0}-5 \beta ^2 a_{1,0} f_{0,0}\right)+\sigma _0^2 a_{1,1} f_{0,0}\right)} \) \lam^2 , \\
\sig_{0,2}
   &=   \frac{\tau}{24 \sigma _0^5} \( {12 \beta ^2 \sigma _0^4 a_{0,2} b_{0,0}-4 \beta ^4 \sigma _0^2 \left(2 a_{0,1}^2 b_{0,0}+a_{0,1} f_{0,0} f_{0,1}+a_{0,2} f_{0,0}^2\right)+9 \beta ^6 a_{0,1}^2 f_{0,0}^2} \) \\ & \qquad
                + \frac{\tau^2}{24 \sigma _0^3} \( {\beta^2 \left(\sigma _0^2 \left(-2 \beta ^2 a_{0,1}^2 b_{0,0}+a_{0,1} \left(2 c_{0,0}+\beta  f_{0,0}\right) \left(2 c_{0,1}+\beta  f_{0,1}\right)+a_{0,2} \left(2 c_{0,0}+\beta  f_{0,0}\right){}^2\right)-3 \beta ^2 a_{0,1}^2 c_{0,0} \left(c_{0,0}+\beta  f_{0,0}\right)\right)} \) \\[-1.5em] & \qquad
                + \frac{\tau}{24 \sigma _0^5} \( {\beta ^3 \left(-9 \beta ^2 a_{0,1}^2 f_{0,0} \left(2 c_{0,0}+\beta  f_{0,0}\right)+4 \sigma _0^2 a_{0,2} f_{0,0} \left(2 c_{0,0}+\beta  f_{0,0}\right)+4 \sigma _0^2 a_{0,1} \left(f_{0,1} \left(c_{0,0}+\beta  f_{0,0}\right)+c_{0,1} f_{0,0}\right)\right)} \) \lam \\[-1.5em] & \qquad
                +   \frac{1}{12 \sigma _0^7} \( {\beta ^4 \left(2 \sigma _0^2 \left(2 a_{0,1}^2 b_{0,0}+a_{0,1} f_{0,0} f_{0,1}+a_{0,2} f_{0,0}^2\right)-9 \beta ^2 a_{0,1}^2 f_{0,0}^2\right)} \) \lam^2 .
\end{align}
\end{small}

\subsection{Comparison to other implied volatility expansions}
\label{sec:compare}
As previously mentioned, when $\beta = 1$, options written on the $\log$ LETF $Z$ are equivalent to options written the $\log$ ETF $X$.  In this special case, the implied volatility expansion discussed in this manuscript reduces to the implied volatility expansion developed in \cite{lorig-pagliarani-pascucci-2}.  If one additionally chooses $(\xb,\yb) = (x,y)$, then the implied volatility approximation given in \cite{lorig-pagliarani-pascucci-2} is equivalent to the implied volatility expansion given in \cite{bompisgobet2013}.  However, the expansion presented here and in \cite{lorig-pagliarani-pascucci-2} is derived using PDE methods, whereas the expansion presented in \cite{bompisgobet2013} is developed using tools from Malliavin calculus.  As of yet, the implied volatility approximation of \cite{bompisgobet2013} have not been extended to options on LETFs.
\par
We note that, for options written on the ETF $X$, extensive comparisons to other implied volatility expansions have been carried out in \cite{lorig-pagliarani-pascucci-2}.  In particular, for the Heston model, the approximation method presented here is compared to the approximation method in \cite{forde-jacquier-lee}, for CEV, it is compared to the approximation method of \cite{hagan-woodward}, and for SABR it is compared to the approximation of \cite{sabr}.  However, neither \cite{forde-jacquier-lee}, \cite{hagan-woodward} nor \cite{sabr} develop approximations for implied volatilities written on the LETF $Z$, as we do here.
\par
Two other methods one might conceivably use to compute approximate options prices and implied volatilities on LETFs are the \emph{heat kernel} and \emph{large deviations} methods, which are discussed, for example, in \cite{lorig-forde-1,laborderebook,Gatheral2012}.  Generally speaking, these methods all rely on computing geodesic distances on a Riemannian manifold whose metric is the inverse of the covariance matrix of the underlying diffusion.  It is not clear how one would compute large deviation estimates and geodesic distances for the metric associated with the process $(X,Y,Z)$ since the diffusion matrix is singular (see Remark \ref{rmk:elliptic}).

\subsection{Implied volatility and $\log$-moneyness scaling}
\label{sec:scaling}
Let us continue to work in the time-homogeneous setting.
Let $\sig_Z(\tau,\lam)$ be the implied volatility of a call written on the LETF $Z$ with time to maturity $\tau$ and $\log$-moneyness $\lam=(k-z)$ and let $\sig_X(\tau,\lam)$ be the implied volatility of a call written on the ETF $X$ time to maturity $\tau$ and $\log$-moneyness $\lam=(k-x)$.  The expressions above provide an explicit approximation for $\sig_Z(\tau,\lam)$ and $\sig_X(\tau,\lam)$ in a general time-homogeneous LSV setting (for $\sig_X(\tau,\lam)$, simply set $\beta=1$).  These expressions show the highly non-trivial dependence of the implied volatility on the leverage ratio $\beta$, and  are  useful for the purposes of calibration.   The  implied volatility surfaces
$(\tau,\lam) \mapsto \sig_Z(\tau,\lam)$
and
$(\tau,\lam) \mapsto \sig_X(\tau,\lam)$ can potentially behave very differently.  Nevertheless, for price comparison across leverage ratios,  it would  practical  to relate them, albeit  heuristically or approximately.
To this end, we now introduce some intuitive scalings.  Examining the lowest-order terms $\sig_0$ and $\sig_1$ we observe
\begin{align}
\text{LETF}:&&
\sig_Z
    &\approx
    |\beta|\sqrt{2a_{0,0}}
            + |\beta| \( \frac{a_{1,0}}{2 \sqrt{2a_{0,0}}}+
            \frac{a_{0,1} f_{0,0}}{2 (2a_{0,0})^{3/2}} \) \frac{\lam}{\beta} + \Oc(\tau) ,
            \label{eq:letf} \\
\text{ETF}:&&
\sig_X
    &\approx
    \sqrt{2a_{0,0}}
            + \( \frac{ a_{1,0}}{2 \sqrt{2a_{0,0}}}+
            \frac{ a_{0,1} f_{0,0}}{2 (2a_{0,0})^{3/2}} \) \lam + \Oc(\tau) . \label{eq:etf}
\end{align}
Comparing $\sig_Z$ with $\sig_X$, we see two effects from the leverage ratio $\beta$.  First, the vertical axis of $\sig_Z$ is scaled by a factor of $|\beta|$.  Second, the horizontal axis is scaled by a factor of $1/\beta$.  In particular, this means that if $\beta<0$ the slopes of $\sig_X$ and $\sig_Z$ will have opposite signs.
For small $\tau$ the contribution of the $\Oc(\tau)$ terms in the expansion will be insignificant.  In light of the above observations, it is natural to introduce $\sig_X^{(\beta)}$ and $\sig_Z^{(1/\beta)}$, the \emph{scaled implied volatilities}, which we define as
\begin{align}
\sig_X^{(\beta)}(\tau,\lam)
    &:= |\beta|\sig_X(\tau,\lam/\beta) , &
\sig_Z^{(1/\beta)}(\tau,\lam)
  &:= \frac{1}{|\beta|}\sig_Z(\tau,\beta \, \lam) . \label{eq:scaling}
\end{align}
These definitions offer two ways to link the implied volatilities surfaces $\sig_X$ and $\sig_Z$.  Viewed one way, the ETF implied volatility $\sig_X(\tau,\lam)$ should roughly coincide with the LETF implied volatility $\frac{1}{|\beta|}\sig_Z(\tau,\beta \lam)$.  Conversely, the LETF implied volatility $\sig_Z(\tau,\lam)$ should be close to the ETF implied volatility $|\beta|\sig_X(\tau,\lam/\beta)$.  In other words, from \eqref{eq:letf}, \eqref{eq:etf} and \eqref{eq:scaling}, we see that for small $\tau$
\begin{align}
\sig_Z(\tau,\lam)
    &\approx \sig_X^{(\beta)}(\tau,\lam), &
\sig_X(\tau,\lam)
    &\approx \sig_Z^{(1/\beta)}(\tau,\lam). \label{eq:sigX.sigZ}
\end{align}
In Figure \ref{fig:scaledIV}, using empirical options data from the S\&P500-based ETF and LETFs, we plot $\sig_Z$ and $\sig_Z^{(1/\beta)}$, the unscaled and scaled implied volatilities, respectively.  The figure demonstrates the pronounced effect  of the scaling argument. Prior to scaling (left panel), the    implied volatilities of the LETFs, SSO ($\beta = +2$) and SDS ($\beta = -2$), have much higher values than those of  the unleveraged ETF SPY $(\beta=+1)$. Moreover, the SDS implied volatility is increasing in log-moneyness. After scaling the LETF implied volatilities  according to  \eqref{eq:scaling} (right panel), they are  brought very close to the ETF implied volatility and they are now all   downward sloping. In Section \ref{sec:examples}, we will compute explicit approximations for $\sig_X(\tau,\lam)$ and $\sig_Z^{(1/\beta)}(\tau,\lam)$ for three well-known models: CEV, Heston and SABR.  As we shall see, although these three models induce distinct implied volatility surfaces, for small $\tau$ the role of $\beta$ in relating $\sig_X$ to $\sig_Z$ will be captured by \eqref{eq:sigX.sigZ}.
\par
We emphasize, however, that the scaling alone is not sufficient to capture the complexity of the LETF implied volatility surface.  Indeed, as $\tau$ increases, we expect $\sig_Z^{(1/\beta)}$ to diverge from $\sig_X$.  This discrepancy is due to the integrated variance contribution to the terminal value of $Z$, as can be seen from \eqref{eq:ZinX}.  Thus, for longer maturities, an accurate approximation of the LETF implied volatility surface must include higher terms in $\tau$.  From the general implied volatility expression, we can see that the role of $\beta$ in the $\Oc(\tau)$ terms is complicated and does not lend itself to a simple scaling argument.  For this reason the full implied volatility expansion  -- not just the scaling argument -- is important.

\begin{figure}
\centering
\includegraphics[width=.475\textwidth]{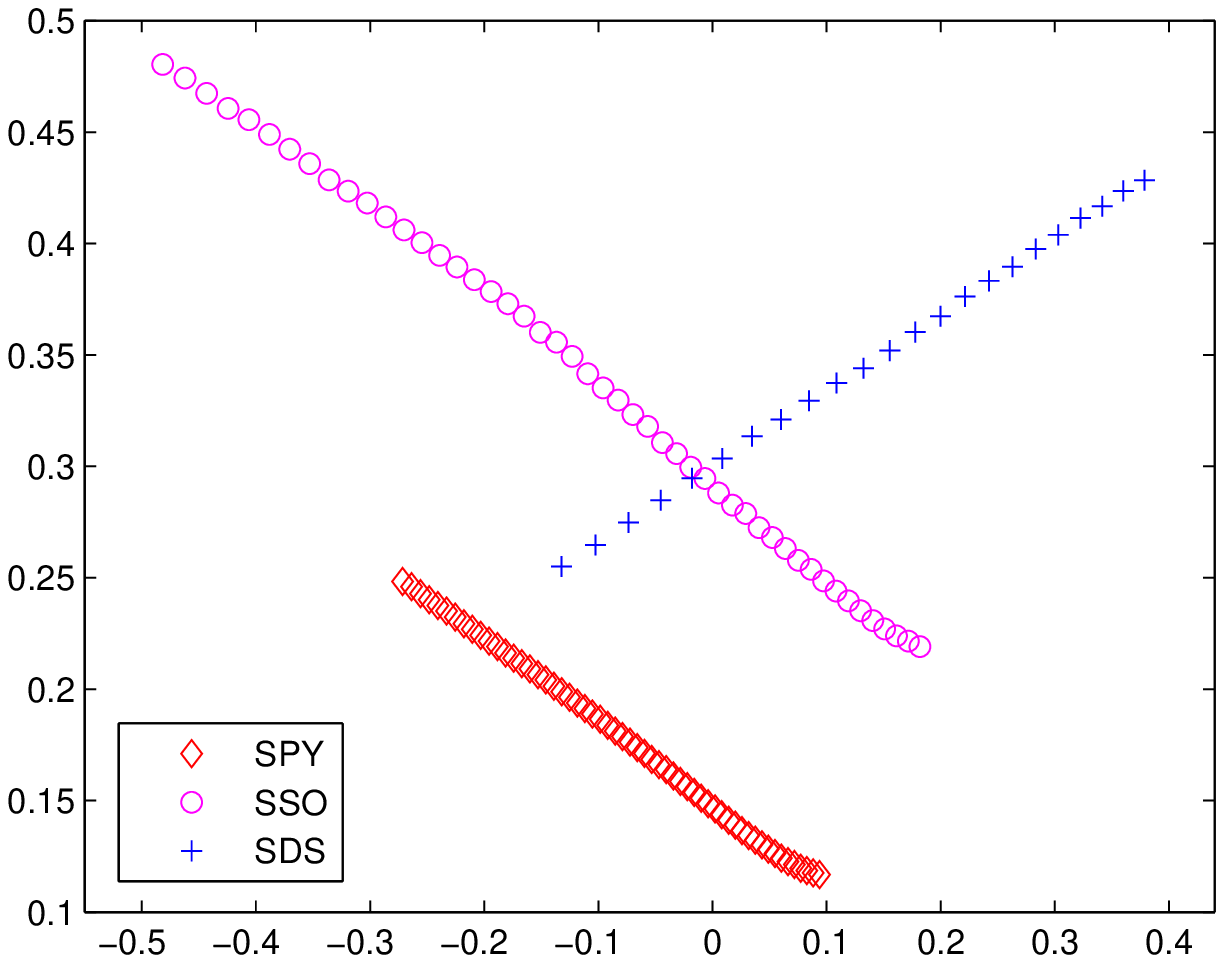}
 \includegraphics[width=.475\textwidth ]{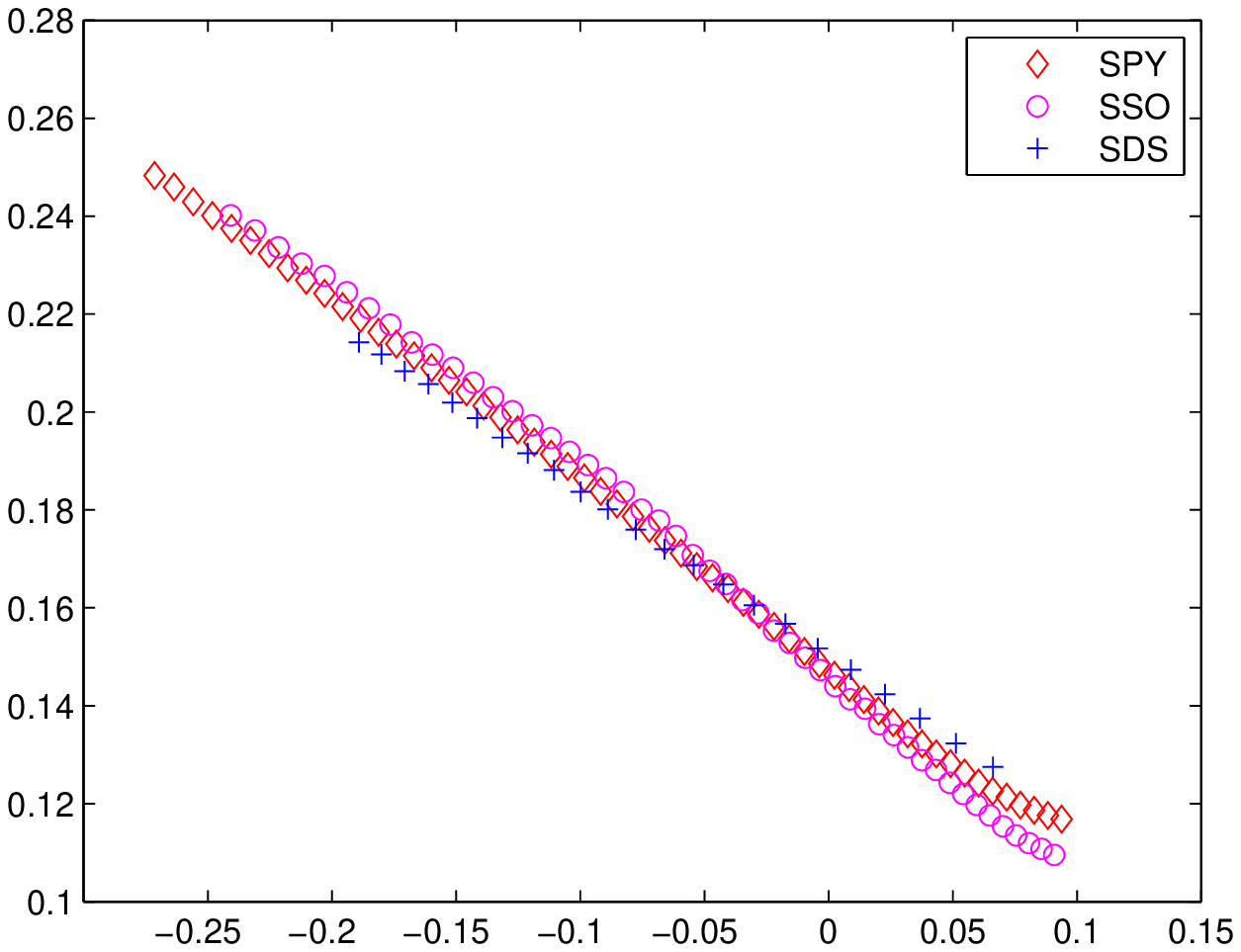}
\caption{\small{Left: {Empirical implied volatilities $\sig_Z(\tau,\lam)$ plotted as a function of log-moneyness $\lam$ for SPY (red diamonds, $\beta=+1$), SSO (purple circles, $\beta=+2$), and SDS (blue crosses, $\beta=-2$) on August 15, 2013 with $\tau=155$ days to maturity.  Note that the implied volatility of SDS is increasing in the LETF log-moneyness. Right: Using the same data, the scaled LETF implied volatilities $\sig_Z^{(1/\beta)}(\tau,\lam)$ nearly coincide.}}}
\label{fig:scaledIV}
 \end{figure}

\begin{remark} A recent paper by \cite{leungsircarLETF} postulates an alternative implied volatility scaling based on stochastic arguments.  Given the terminal ETF value $X_{\tau} = k$, they compute the expected future $\log$-moneyness $Z_\tau-z$
\begin{align}
\Eb_{x,y,z} [ Z_T - z | X_T = k ]
    &=  \beta ( k - x ) - \frac{1}{2}\beta ( \beta - 1 )
            \int_0^\tau \Eb_{x,y,z}[ \sig^2(s,X_s,Y_s) | X_\tau = k ] \, \dd s , \label{eq:lm.scaling}
\end{align}
where $\Eb_{x,y,z}[\cdot]=\Eb[\cdot|X_0=x,Y_0=y,Z_0=z]$. They also note, from the ETF and LETF SDEs that the volatility of $Z$ is $|\beta|$ times the volatility of $X$.  Using the above as heuristic, the authors propose to scale implied volatilities as follows
\begin{align}
\sig_Z(\tau,\lam)
    &=  |\beta|\sig_X(\tau,\beta \lam - \frac{1}{2}\beta(\beta-1)I(\tau)) , \\
I(\tau)
    &=  \int_0^\tau \Eb_{x,y,z}[ \sig^2(s,X_s,Y_s) | X_\tau = k ] \, \dd s .
\end{align}
In \cite{leungsircarLETF}, the value of $I(\tau)$ is estimated using an average from observed implied volatility.  In contrast, the scaling proposed in \eqref{eq:scaling} does not attempt to account for the integral in \eqref{eq:lm.scaling}.  Nevertheless, the effect of the integrated variance is captured by $\Oc(\tau)$ terms in the general implied volatility expansion.
\end{remark}

%
%

\section{Examples}
\label{sec:examples}
In this Section, we provide explicit expressions for implied volatilities under three different model dynamics: CEV, Heston and SABR.
Special attention will be paid to the role of $\beta$, the leverage ratio.  In the examples below, we fix $(\xb(\cdot),\yb(\cdot))=(X_0,Y_0)$ and we evaluate implied volatilities at time $t=0$ and maturing at time $T=\tau$.
\par
We note that, although {Theorem} \ref{cor:u} establishes the order of accuracy of our pricing approximation as $\tau \to 0$, our numerical tests indicate that the implied volatility expansion gives an accurate approximation of $\sig_Z^{(1/\beta)}$ for maturities of multiple years.  Nevertheless, options on LETFs only currently trade with maturities of less than 1.25 years\footnote{Delayed quotes for ETF and LETF options of all traded maturities are available on the CBOE and Yahoo Finance websites.}.\,  \cite{leungsircarLETF} have plotted the empirical implied volatilities  for four  S\&P500 based LETF options ($\beta=\pm 2,\pm 3$),  all with maturities of less than a year. Hence, in the numerical examples below, we focus on these maturities.


\subsection{CEV}
\label{sec:cev}
In the Constant Elasticity of Variance (CEV) local volatility model of \cite{CoxCEV}, the dynamics of the underlying $S$ are given by
\begin{align}
\dd S_t
    &=  \del S_t^{\gamma-1} S_t \dd W_t^x , &
S_0
    &> 0,
\end{align}
where, to preserve the martingale property of the process $S$ (cf. \cite{heston2007}), the
parameter $\gamma$ is assumed to be less than or equal to $1$. The dynamics of $(X,Z)=(\log S,
\log L)$ are
\begin{align}
\dd X_t
    &=  -\frac{1}{2} \del^2 \ee^{2(\gamma-1)X_t} \dd t + \del \, \ee^{(\gamma-1)X_t} \dd W_t^x , &
X_0
    &=  x := \log S_0 . \\
\dd Z_t
    &=  -\frac{1}{2} \beta^2 \del^2 \ee^{2(\gamma-1)X_t} \dd t + \beta \del \, \ee^{(\gamma-1)X_t} \dd W_t^x , &
Z_0
    &=  z := \log L_0 .
\end{align}
The generator of $(X,Z)$ is given by
\begin{align}
\Ac
    &=  \frac{1}{2} \del^2 \ee^{2(\gamma-1)x} \(
                ( \d_x^2 - \d_x ) + \beta^2 ( \d_z^2 - \d_z ) + 2 \beta \d_x \d_z \) .
\end{align}
Thus, from \eqref{eq:A}, we identify
\begin{align} \label{CEVatof}
a(x,y)
    &=  \frac{1}{2} \del^2 \ee^{2(\gamma-1)x} , &
b(x,y)
    &=  0 , &
c(x,y)
    &=  0 , &
f(x,y)
    &=  0 .
\end{align}
Using equations \eqref{eq:sig.n}, \eqref{eq:d.sigma} and \eqref{eq:Un} we compute
\begin{align}
\sig_0
    &=  |\beta|\sqrt{\ee^{2 x (\gamma-1 )} \del^2},\label{eq:cevsig0} \\
\sig_{1}
   &=   \tau \( \frac{(\beta -1) (\gamma -1) \sigma _0^3}{4 \beta ^2} \)
                + \( \frac{(\gamma -1) \sigma _0}{2 \beta } \) (k-z) , \\
\sig_{2}
   &=   \tau \( \frac{(\gamma -1)^2 \sigma _0^3}{24 \beta ^2} \)
                + \tau^2 \( \frac{(2 \beta  (6 \beta -13)+13) (\gamma -1)^2 \sigma _0^5}{96 \beta ^4} \) \\ & \qquad
                + \tau \( \frac{7 (\beta -1) (\gamma -1)^2 \sigma _0^3}{24 \beta ^3} \) (k-z)
                +   \( \frac{(\gamma -1)^2 \sigma _0}{12 \beta ^2} \) (k-z)^2 , \\
\sig_{3}
   &=  \tau^2 \( \frac{5 (\beta -1) (\gamma -1)^3 \sigma _0^5}{32 \beta ^4} \)
                + \tau^3 \( \frac{(\beta -1) \left(26 \beta ^2-70 \beta +35\right) (\gamma -1)^3 \sigma _0^7}{384 \beta ^6} \) \\ & \qquad
                + \tau \( \frac{(\gamma -1)^3 \sigma _0^3}{16 \beta ^3} \) (k-z)
                +   \tau^2 \( \frac{5 (2 \beta  (4 \beta -9)+9) (\gamma -1)^3 \sigma _0^5}{192 \beta ^5} \) (k-z) \\ & \qquad
                +   \tau \( \frac{7 (\beta -1) (\gamma -1)^3 \sigma _0^3}{48 \beta ^4} \) (k-z)^2 .
\end{align}
We observe that the factor $(\gamma - 1)$ appears in every term of  these expressions. In particular, when $\gamma=1$, $\sig_0 = |\beta|\delta$ and $\sig_1=\sig_2=\sig_3= 0$. The higher order terms also vanish since $a(x,y)=  \frac{1}{2} \del^2$ in this case (see \eqref{CEVatof}). Hence, just as in the Black-Scholes case, the implied volatility expansion becomes flat, as expected.

In Figure \ref{fig:cev.iv} we plot our third-order approximation of the scaled implied volatility $\sig_Z^{(1/\beta)}(\tau,\lam)$ in the CEV model with leverages $\beta=\{+2,-2\}$ and with maturities {$\tau=\{0.25,0.5,1\}$} years.  For comparison, we also plot the exact scaled implied volatility $\sig_Z^{(1/\beta)}(\tau,\lam)$ and the exact implied volatility of the ETF $\sig_X(\tau,\lam)$.
The exact scaled implied volatility $\sig_Z^{(1/\beta)}$ of the LETF is computed by obtaining call prices by Monte Carlo simulation and then by inverting the Black-Scholes formula numerically.
The exact implied volatility $\sig_X(\tau,\lam)$ of the ETF is computed using the exact call price formula, available in \cite{CoxCEV}, and then inverting the Black-Scholes formula numerically.


\subsection{Heston}
\label{sec:heston}
In the Heston model, due to \cite{heston1993}, the dynamics of the underlying $S$ are given by
\begin{align}
\dd S_t
    &=  \sqrt{V_t} S_t \dd W_t^x , &
S_0
    &>  0 , \\
\dd V_t
    &=  \kappa (\theta - V_t) \dd t + \del \sqrt{V_t} \dd W_t^y , &
V_0
    &> 0 , \\
\dd\<W^x,W^y\>_t
    &=  \rho \, \dd t .
\end{align}
In $\log$ notation $(X,Y,Z) := (\log S, \log V, \log L)$ we have the following dynamics
\begin{align}
\begin{aligned}
\dd X_t
    &=  -\frac{1}{2} \ee^{Y_t} \dd t + \ee^{\tfrac{1}{2}Y_t} \dd W_t^x , &
X_0
    &=  x := \log S_0 , \\
\dd Y_t
    &=  \( (\kappa \theta -\tfrac{1}{2}\del^2) \ee^{-Y_t} - \kappa \) \dd t + \del \, \ee^{-\tfrac{1}{2}Y_t} \dd W_t^y , &
Y_0
    &=  y := \log V_0 , \\
\dd Z_t
    &=  - \beta^2 \frac{1}{2} \ee^{Y_t} \dd t + \beta \ee^{\tfrac{1}{2}Y_t} \dd W_t^x , &
Z_0
    &=  z := \log L_0 , \\
\dd\<W^x,W^y\>_t
    &=  \rho \, \dd t .
\end{aligned} \label{eq:model.Heston}
\end{align}
The generator of $(X,Y,Z)$ is given by
\begin{align}
\Ac
    &=  \frac{1}{2} \ee^{y} \( (\d_x^2 - \d_x) + \beta^2 (\d_z^2 - \d_z) + 2 \beta \d_x \d_z \) \\ &\qquad
                + \( (\kappa \theta -\tfrac{1}{2}\del^2) \ee^{-y} - \kappa \) \d_y
        + \frac{1}{2} \del^2 \ee^{-y} \d_y^2 + \rho \, \del \( \d_x \d_y + \beta \d_x \d_z \) .
\end{align}
Thus, from \eqref{eq:A}, we identify
\begin{align}
a(x,y)
    &=  \frac{1}{2} \ee^{y} , &
b(x,y)
    &=  \frac{1}{2} \del^2 \ee^{-y} , &
c(x,y)
    &=  \( (\kappa \theta -\tfrac{1}{2}\del^2) \ee^{-y} - \kappa \) , &
f(x,y)
    &=  \rho \, \del .
\end{align}
Using equations \eqref{eq:sig.n}, \eqref{eq:d.sigma} and \eqref{eq:Un} we obtain
\begin{align}
\sig_0
    &=  |\beta|\sqrt{\ee^y}, \label{eq:hestonsig0} \\
\sig_{1}
   &=   \frac{\tau}{8 \sigma _0} \( {\sigma _0^2 (\beta  \delta  \rho -2 \kappa )-\beta ^2 \left(\delta ^2-2 \theta  \kappa \right)} \)
                + \frac{1}{4 \sigma _0} \( {\beta  \delta  \rho } \) (k-z) , \\
\sig_{2}
   &=   \frac{\tau}{96 \sigma _0} \( {\beta ^2 \delta ^2 \left(\rho ^2+8\right)} \) \\ & \qquad
                + \frac{\tau^2}{384 \sigma _0^3} \( {-3 \beta ^4 \left(\delta ^2-2 \theta  \kappa \right)^2-2 \beta ^2 \sigma _0^2 \left(\delta ^2-2 \theta  \kappa \right) (\beta  \delta  \rho -2 \kappa )+4 \sigma _0^4 \left(\beta  \delta  \left(\beta  \delta  \left(2 \rho ^2-1\right)-5 \kappa  \rho \right)+5 \kappa ^2\right)} \) \\[-1.5em] & \qquad
                + \frac{\tau}{96 \sigma _0^3} \( {\beta  \delta  \rho  \left(5 \beta ^2 \left(\delta ^2-2 \theta  \kappa \right)+\sigma _0^2 (2 \kappa -\beta  \delta  \rho )\right)} \) (k-z)
                +   \frac{1}{48 \sigma _0^3} \( {\beta ^2 \delta ^2 \left(2-5 \rho ^2\right)} \) (k-z)^2 , \\
\sig_{3}
   &=  \frac{\tau^2}{768 \sigma _0^3} \( {\beta ^2 \delta ^2 \left(\beta ^2 \left(5 \rho ^2+4\right) \left(\delta ^2-2 \theta  \kappa \right)+3 \rho ^2 \sigma _0^2 (\beta  \delta  \rho -2 \kappa )\right)} \) \\ & \qquad
                + \frac{\tau^3}{3072 \sigma _0^5} \( {-3 \beta ^6 \left(\delta ^2-2 \theta  \kappa \right)^3+\beta ^4 \sigma _0^2 \left(\delta ^2-2 \theta  \kappa \right)^2 (\beta  \delta  \rho -2 \kappa )+4 \beta ^2 \kappa  \sigma _0^4 \left(\delta ^2-2 \theta  \kappa \right) (\beta  \delta  \rho -\kappa )} \) \\ & \qquad
                + \frac{\tau^3}{3072 \sigma _0^5} \( {2 \sigma _0^6 (\beta  \delta  \rho -2 \kappa ) \left(\beta  \delta  \left(\beta  \delta  \left(5 \rho ^2-6\right)-6 \kappa  \rho \right)+6 \kappa ^2\right)} \) \\ & \qquad
                + \frac{\tau}{384 \sigma _0^3} \( -{\beta ^3 \delta ^3 \rho  \left(9 \rho ^2+8\right)} \) (k-z) \\ & \qquad
                +   \frac{\tau^2}{1536 \sigma _0^5} \( {\beta  \delta  \rho  \left(21 \beta ^4 \left(\delta ^2-2 \theta  \kappa \right)^2-10 \beta ^2 \sigma _0^2 \left(\delta ^2-2 \theta  \kappa \right) (\beta  \delta  \rho -2 \kappa )\right)} \) (k-z) \\ & \qquad
                +   \frac{\tau^2}{1536 \sigma _0^5} \( {\beta  \delta  \rho  \left(4 \sigma _0^4 \left(\beta  \delta  \left(\beta  \left(\delta -2 \delta  \rho ^2\right)+3 \kappa  \rho \right)-3 \kappa ^2\right)\right)} \) (k-z) \\ & \qquad
                +   \frac{\tau}{384 \sigma _0^5} \( -{\beta ^2 \delta ^2 \left(\beta ^2 \left(23 \rho ^2-8\right) \left(\delta ^2-2 \theta  \kappa \right)+\left(7 \rho ^2-2\right) \sigma _0^2 (2 \kappa -\beta  \delta  \rho )\right)} \) (k-z)^2 \\ & \qquad
                +   \frac{1}{96 \sigma _0^5} \( {\beta ^3 \delta ^3 \rho  \left(8 \rho ^2-5\right)} \) (k-z)^3 .  \label{eq:hestonsig2}
\end{align}
For longer maturities, the accuracy of the implied volatility expansion can be improved by choosing a time-dependent expansion point for the $Y$ process: $\yb(t) = \Eb_y Y_t$.  In this case, the formulas for $\sig_0$, $\sig_1$, $\sig_2$ and $\sig_3$ remain explicit.  However, as the expressions are quite long, we omit them.
\par
\cite{haughLETF} noticed from the SDEs that when $X$ has Heston dynamics with parameters ($\kappa$, $\theta$, $\del$, $\rho$, $y$), then $Z$ has Heston dynamics with parameters
\begin{align} \label{hestonz}
(\kappa_Z, \theta_Z, \del_Z, \rho_Z, y_Z)
    &=  (\kappa, \beta^2 \theta, |\beta|\del, \text{sign}(\beta) \rho, y + \log \beta^2) .
\end{align}
The characteristic function of $X_\tau$ is computed explicitly in \citet*{heston1993} and \cite*{Bakshi1997}
\begin{align}
\eta_X(\tau,x,y,\xi)
    &:= \log \Eb[ \ee^{ \ii \xi X_\tau} | X_0 = x, Y_0 = y ]
    =  {\ii \xi x + C(\tau,\xi) + D(\tau,\xi) \ee^y} , \label{eq:heston.char} \\
C(\tau,\xi)
    &=  \frac{\kappa \theta}{ \del^2} \( (\kappa - \rho \delta  \ii \xi + d(\xi) ) \tau
            -2 \log \[ \frac{1-f(\xi) \ee^{d(\xi)\tau}}{1-f(\xi)}\]\) , \\
D(\tau,\xi)
    &=  \frac{\kappa - \rho \del \ii \xi + d(\xi)}{\del^2} \frac{1-\ee^{d(\xi)\tau}}{1-f(\xi) \ee^{d(\xi)\tau}} , \\
f(\xi)
    &=  \frac{\kappa - \rho \del \ii \xi + d(\xi)}{\kappa - \rho \del \ii \xi - d(\xi)} , \\
d(\xi)
    &=  \sqrt{ \del^2 (\xi^2 + \ii \xi) + (\kappa - \rho \ii \xi \del)^2} .
\end{align}
Since $Z$ also has Heston dynamics, the characteristic function of $Z_\tau$ follows directly
\begin{align}
\eta_Z(\tau,z,y,\xi)
    &:= \log \Eb[ \ee^{ \ii \xi Z_\tau} | Y_0 = y , Z_0 = z]
    =       \eta_X(\tau,z,y,\xi) &
    &\text{with}&
(\kappa, \theta, \del, \rho,y)
    &\to    (\kappa_Z, \theta_Z, \del_Z, \rho_Z, y_Z) .
\end{align}
The price of a European call option with payoff $\phi(z)=(\ee^z-\ee^k)^+$ can then be computed using standard Fourier methods
\begin{align}
u^\text{Hes}(\tau,z,y)
    &=  \frac{1}{2\pi} \int_\Rb \dd \xi_r \, \ee^{\eta_Z(\tau,z,y,\xi)} \phih(\xi) , &
\phih(\xi)
    &=  \frac{-\ee^{k-\ii k \xi}}{ \ii \xi + \xi^2 } , &
\xi
    &=  \xi_r + \ii \xi_i , &
\xi_i
    &<  -1 . \label{eq:u.Heston}
\end{align}
Note, since the call option payoff $\phi(z)=(\ee^z -\ee^k)^+$ is not in $L^1(\Rb)$, its Fourier
transform $\hh(\xi)$ must be computed in a generalized sense by fixing an imaginary component of the Fourier variable $\xi_i < -1$. Using \eqref{eq:u.Heston} the exact implied volatility $\sig$ can be computed to solving \eqref{eq:imp.vol.def} numerically.
\par
{Moreover, it is worth noting that relationship \eqref{hestonz} can be inferred from our implied volatility expressions. Indeed, the dependence on $\beta$ in expansions \eqref{eq:hestonsig0}-\eqref{eq:hestonsig2}  is present \textit{only} in  the terms  $\beta^2 \theta, |\beta|\del, \text{sign}(\beta) \rho, y + \log \beta^2$ of the coefficients.  For instance, we can write the zeroth-order term   $\sig_0 = \sqrt{e^{ y + \log \beta^2}} = \sqrt{e^{y_Z}}$, and  the   coefficient of $(k-z)$ in $\sig_1$   is  $\beta \delta \rho/4 \sig_0 = |\beta| \delta \text{sign}(\beta)\rho/4 \sig_0 = \delta_Z \rho_Z/4\sig_0$, as per the notations in  \eqref{hestonz}. Similar verification procedures for other terms confirm relationship \eqref{hestonz}. }
\par
In Figure \ref{fig:heston.iv} we plot our third-order approximation of the scaled implied volatility $\sig_Z^{(1/\beta)}(\tau,\lam)$ in the Heston model with leverages $\beta=\{+2,-2\}$ and with maturities {$\tau=\{0.25,0.5,1\}$} years.
{For the longest maturity $\tau=1$, we use the implied volatility expansion corresponding to $\yb(t) = \Eb_y Y(t)$.}
For comparison, we also plot the exact scaled implied volatility $\sig_Z^{(1/\beta)}(\tau,\lam)$ and the exact implied volatility of the ETF $\sig_X(\tau,\lam)$.
The exact scaled implied volatility $\sig_Z^{(1/\beta)}$ of the LETF is computed by obtaining call prices from \eqref{eq:u.Heston} and then by inverting the Black-Scholes formula numerically.
The exact implied volatility $\sig_X(\tau,\lam)$ of the ETF is computed in the same manner.


\subsection{SABR}
\label{sec:sabr}
The SABR model of \citet*{sabr} is a local-stochastic volatility model in which the risk-neutral dynamics of $S$ are given by
\begin{align}
\dd S_t
    &=  V_t S_t^{\gamma-1} S_t \dd W_t^x , &
S_0
    &> 0, \\
\dd V_t
    &=  \del V_t \dd W_t^y , &
V_0
    &> 0, \\
\dd\<W^x,W^z\>_t
    &=  \rho \, \dd t .
\end{align}
In $\log$ notation $(X,Y,Z) := (\log S, \log V, \log L)$ we have, we have the following
dynamics:
\begin{align}
\begin{aligned}
\dd X_t
    &=  -\frac{1}{2} \ee^{2Y_t + 2(\gamma-1)X_t} \dd t + \ee^{Y_t + (\gamma-1)X_t} \dd W_t^x , &
X_0
    &=  x := \log S_0 , \\
\dd Y_t
    &=  -\frac{1}{2} \del^2 \dd t + \del \, \dd W_t^y , &
Y_0
    &=  y := \log V_0 , \\
\dd Z_t
    &=  -\frac{1}{2} \beta^2 \ee^{2 Y_t + 2(\gamma-1)X_t} \dd t + \beta \ee^{Y_t + (\gamma-1)X_t} \dd W_t^x , &
Z_0
    &=  z := \log L_0 , \\
\dd\<W^x,W^y\>_t
    &=  \rho \, \dd t .
\end{aligned} \label{eq:model.SABR}
\end{align}
The generator of $(X,Y,Z)$ is given by
\begin{align}
\Ac
    &=  \frac{1}{2} \ee^{2y + 2(\gamma-1)x}
                \( (\d_x^2 - \d_x) + \beta^2(\d_z^2 - \d_z) + 2 \beta \d_x \d_y \) \\ & \qquad
         - \frac{1}{2} \del^2 \d_y + \frac{1}{2} \del^2 \d_y^2
                    + \rho \, \del \, \ee^{y + (\gamma-1)x} ( \d_x \d_y + \beta \d_y \d_z ).
\end{align}
Thus, using \eqref{eq:A}, we identify
\begin{align}
a(x,y)
    &= \frac{1}{2} \ee^{2y + 2(\gamma-1)x}, &
b(x,y)
    &=  \frac{1}{2} \del^2 , &
c(x,y)
    &=  - \frac{1}{2} \del^2 , &
f(x,y)
    &=  \rho \, \del \, \ee^{y + (\gamma-1)x} . &
\end{align}
Using equations \eqref{eq:sig.n}, \eqref{eq:d.sigma} and \eqref{eq:Un} we compute
\begin{align}
\sig_0
    &=  |\beta|\sqrt{\ee^{2 y+2 x (-1+\gamma )} }, &
\sig_1
    &=  \sig_{1,0} + \sig_{0,1} , &
\sig_2
    &=  \sig_{2,0} + \sig_{1,1} + \sig_{0,2} ,
\end{align}
where
\begin{align}
\sig_{1,0}
   &=   \tau \( \frac{(\beta -1) (\gamma -1) \sigma _0^3}{4 \beta ^2} \)
                + \( \frac{(\gamma -1) \sigma _0}{2 \beta } \) (k-z) , \\
\sig_{0,1}
   &=       \tau \( -\frac{1}{4} \delta  \sigma _0 \left(\delta -\rho  \sigma _0 \text{sgn}(\beta )\right) \)
                + \( \frac{1}{2} \delta  \rho  \text{sgn}(\beta ) \) (k-z), \\
\sig_{2,0}
   &=   \tau \( \frac{(\gamma -1)^2 \sigma _0^3}{24 \beta ^2} \)
                + \tau^2 \( \frac{(2 \beta  (6 \beta -13)+13) (\gamma -1)^2 \sigma _0^5}{96 \beta ^4} \) \\ & \qquad
                + \tau \( \frac{7 (\beta -1) (\gamma -1)^2 \sigma _0^3}{24 \beta ^3} \) (k-z)
                +   \( \frac{(\gamma -1)^2 \sigma _0}{12 \beta ^2} \) (k-z)^2 , \\
\sig_{1,1}
   &=   \tau \( \frac{(\gamma -1) \delta  \rho  \sigma _0^2}{12 \left| \beta \right| } \)
                + \tau^2 \( \frac{(\gamma -1) \delta  \sigma _0^3 \left(\beta  (6 \beta -7) \rho  \sigma _0-5 (\beta -1) \delta  \left| \beta \right| \right)}{48 \left| \beta \right| ^3} \) \\ & \qquad
                + \tau \( \frac{(\gamma -1) \delta  \sigma _0 \left(\delta  \left| \beta \right| +(2 \beta -1) \rho  \sigma _0\right)}{24 \beta  \left| \beta \right| } \) (k-z)
                +   \( -\frac{(\gamma -1) \delta  \rho }{3 \left| \beta \right| } \) (k-z)^2 , \\
\sig_{0,2}
   &=   \tau \( \frac{1}{24} \delta ^2 \left(8-3 \rho ^2\right) \sigma _0 \)
                + \tau^2 \( \frac{1}{96} \delta ^2 \sigma _0 \left(5 \delta ^2+4 \left(3 \rho ^2-1\right) \sigma _0^2-\frac{14 \delta  \rho  \sigma _0}{\text{sgn}(\beta )}\right) \) \\ & \qquad
                + \tau \( -\frac{\delta ^2 \rho  \left(\delta -3 \rho  \sigma _0 \text{sgn}(\beta )\right)}{24 \text{sgn}(\beta )} \) (k-z)
                +   \( \frac{\delta ^2 \left(2-3 \rho ^2\right)}{12 \sigma _0} \) (k-z)^2 .
\end{align}
We omit the expression for $\sig_3$ for the sake of brevity.  However, an explicit computations shows that $\sig_3$ contains terms of orders $\tau^2$, $\tau^3$, $\tau (k-z)$, $\tau^2 (k-z)$, $\tau (k-z)^2$ and $(k-z)^3$.
In Figure \ref{fig:sabr.iv} we plot our third-order approximation of the scaled implied volatility $\sig_Z^{(1/\beta)}(\tau,\lam)$ in the SABR model with leverages $\beta=\{+2,-2\}$ and with maturities {$\tau=\{0.25,0.5,1\}$} years.  For comparison, we also plot the exact scaled implied volatility $\sig_Z^{(1/\beta)}(\tau,\lam)$ and the exact implied volatility of the ETF $\sig_X(\tau,\lam)$.
The exact scaled implied volatility $\sig_Z^{(1/\beta)}$ of the LETF is computed by obtaining call prices by Monte Carlo simulation and then by inverting the Black-Scholes formula numerically.
The exact implied volatility $\sig_X(\tau,\lam)$ of the ETF is computed using the exact call price formula, available in \cite{sabr-exact} for the special case $\rho=0$, and then inverting the Black-Scholes formula numerically.

%
%

\section{Conclusion}
\label{sec:conclusion}
In this article, starting from ETF dynamics in a general time-inhomogeneous LSV setting, we derive approximate European-style option prices written on the associated LETFs.  The option price approximation requires only a normal CDF to compute.  Therefore, computational times for prices are comparable to Black-Scholes.  We also establish rigorous error bounds for our pricing approximation.  These error bounds are established through a regularization procedure, which allows us to overcome challenges that arise when dealing with a generator $\Ac(t)$ that is \emph{not} elliptic.
\par
Additionally, we derive an implied volatility expansion that is fully explicit -- polynomial in $\log$-moneyness $\lam=(k-z)$ and (for time-homogeneous models) polynomial in time to maturity.  To aid in the analysis of the implied volatility surface, we discuss some natural scalings of implied volatility.  Furthermore,  we test our implied volatility expansion on three well-known LSV models (CEV, Heston and SABR) and find that the expansion provides an excellent approximation of the true implied volatility.
\par
The markets for  leveraged   ETFs and their options continue to grow,  not only in equities, but also in other sectors such as commodity, fixed-income, and currency. The question of consistent pricing, as we have investigated for equity LETF options in terms implied volatility, is also relevant to LETF options in other sectors.  Naturally, the  valuation of  LETF options will   depend on the dynamics of the LETFs and  underlying price process, which may vary significantly across sectors (see e.g. \cite{GuoLeung14,LeungWard} for commodity LETFs).  Nevertheless, it is both practically and mathematically interesting  to adapt the techniques in the current paper to  investigate the implied volatilities  across leverage ratios  with different underlyings.    From a market stability perspective, it is important for both  investors and regulators  to understand the risks and dependence structure  among ETFs and the price relationships of their traded derivatives.\\

\begin{small}
\noindent \textbf{Acknowledgments}~~~The authors are grateful to  Peter Carr, Emanuel Derman, Martin Haugh, Sebastian Jaimungal, and Ronnie Sircar for their helpful discussions, and we thank  the seminar participants at Morgan Stanley, Columbia University, and Fields Institute for their comments.  The authors also wish to express their gratitude to two anonymous referees {and one anonymous associate editor}, whose comments helped improve the mathematical rigor and clarity of this paper.
\end{small}

%
%

\begin{small}
\bibliographystyle{chicago}
\bibliography{BibTeX-Master-3.0X}
\end{small}

%
%

%
\begin{figure}
\centering
\begin{tabular}{  c  c }
$\tau=0.25$ & $\tau=0.25$ \\
\includegraphics[width=.475\textwidth,height=0.19\textheight]{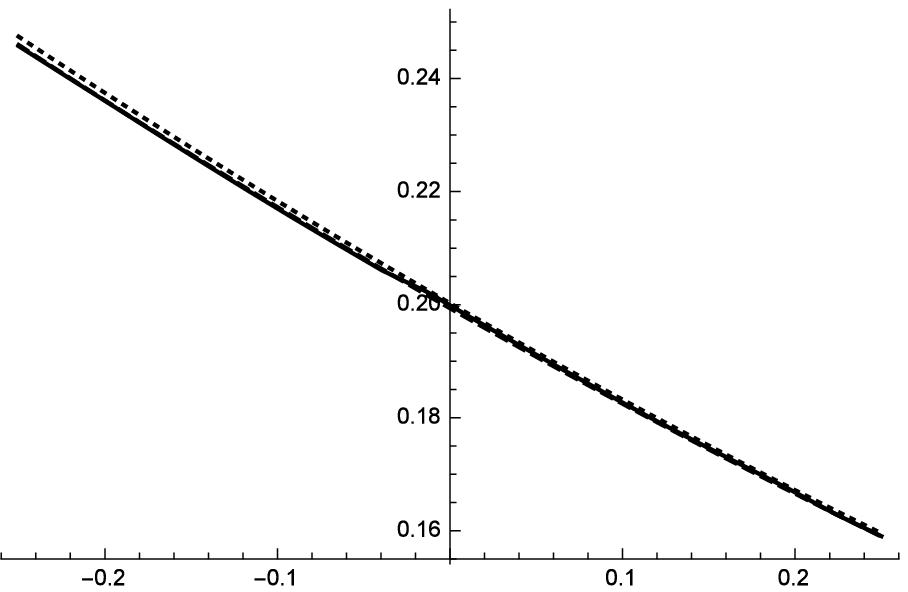}&
\includegraphics[width=.475\textwidth,height=0.19\textheight]{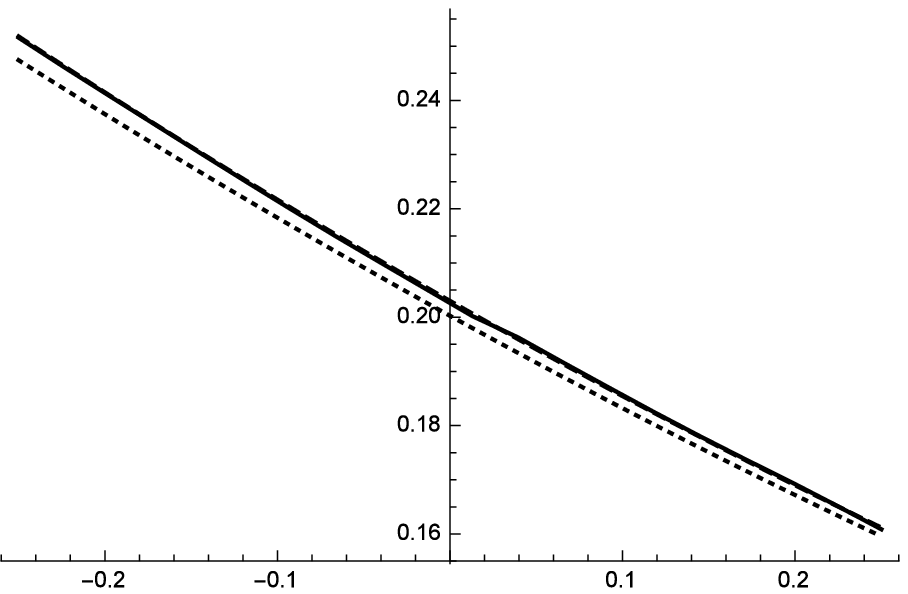}\\
$\tau=0.5$ & $\tau=0.5$ \\
\includegraphics[width=.475\textwidth,height=0.19\textheight]{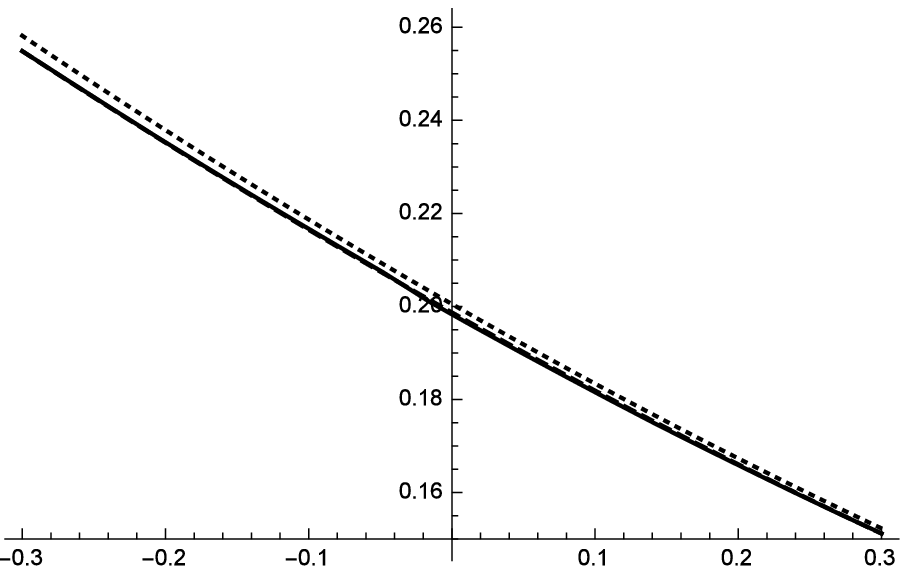}&
\includegraphics[width=.475\textwidth,height=0.19\textheight]{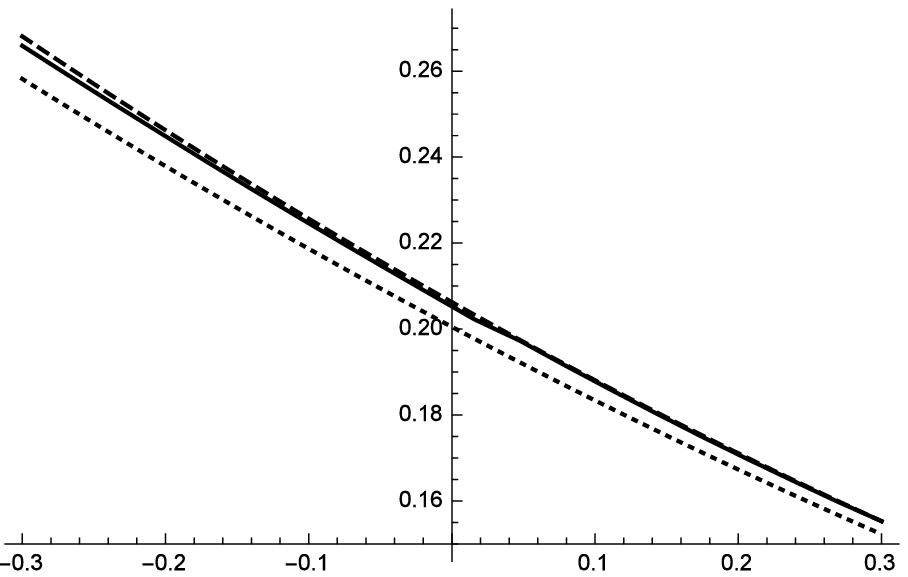}\\
$\tau=1.0$ & $\tau=1.0$ \\
\includegraphics[width=.475\textwidth,height=0.19\textheight]{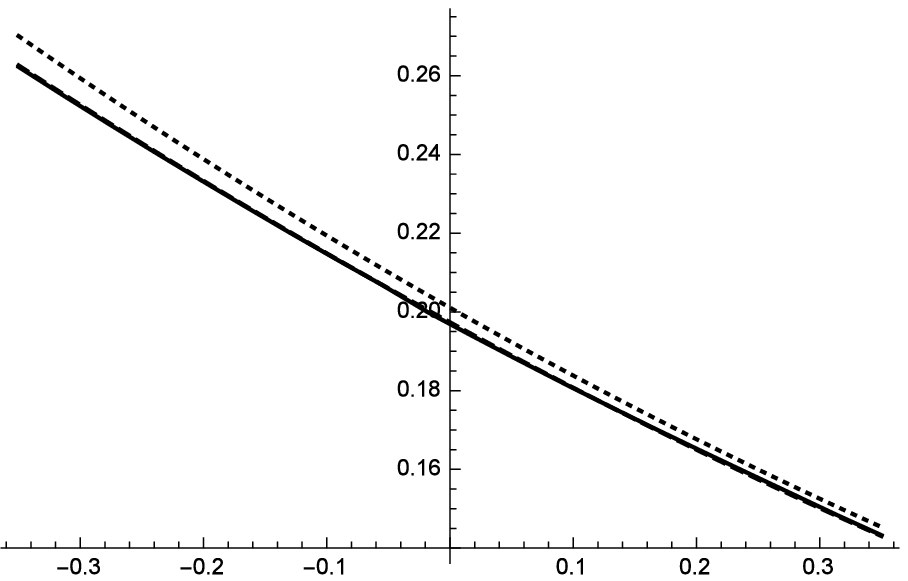}&
\includegraphics[width=.475\textwidth,height=0.19\textheight]{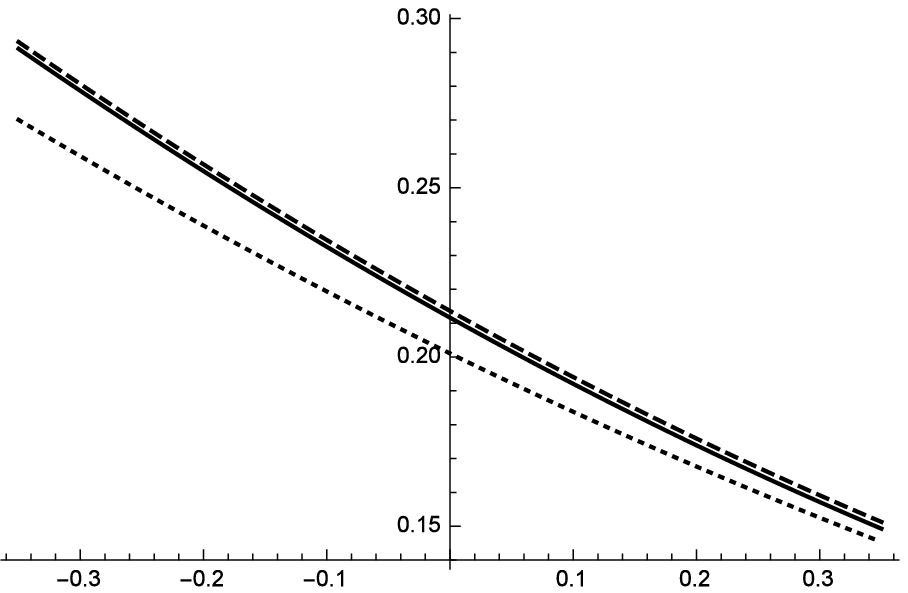}\\
$\beta= + 2$ & $\beta=-2$ \\
\end{tabular}
\caption{\small{
Exact (solid -- computed by Monte Carlo) and approximate (dashed) scaled implied volatility $\sig_Z^{(1/\beta)}(\tau,\lam)$ under CEV model dynamics plotted as a function of log-moneyness  $\lam$.  For comparison, we also plot the exact implied volatility of the CEV model $\sig_Z^{(1)}(\tau,\lam)=\sig_X(\tau,\lam)$ (dotted).  Parameters: $\del=0.2$, $\gam=-0.75$, $x=0$.
For each leverage ratio ($\beta = \pm 2$), as $\tau$ increases, the solid and dotted lines diverge, while the dashed and solid lines {remain so close they are nearly indistinguishable}.}
}
\label{fig:cev.iv}
\end{figure}

%
\clearpage
\begin{figure}
\centering
\begin{tabular}{  c  c }
$\tau=0.25$ & $\tau=0.25$ \\
\includegraphics[width=.475\textwidth,height=0.19\textheight]{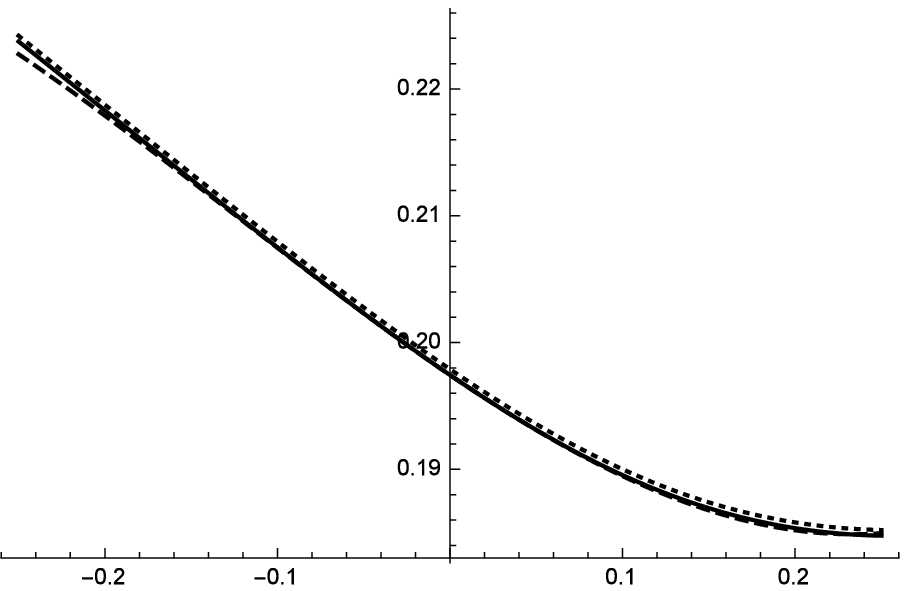}&
\includegraphics[width=.475\textwidth,height=0.19\textheight]{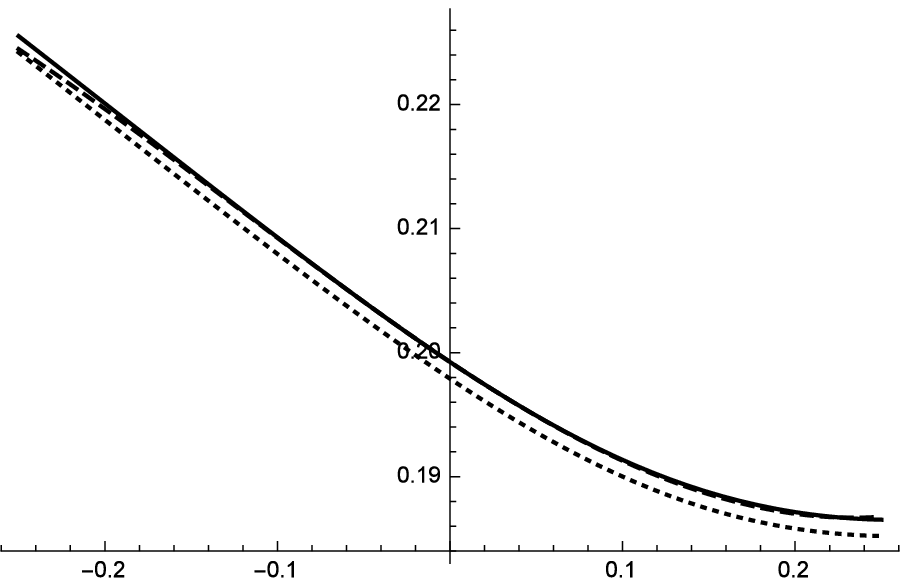}\\
$\tau=0.5$ & $\tau=0.5$ \\
\includegraphics[width=.475\textwidth,height=0.19\textheight]{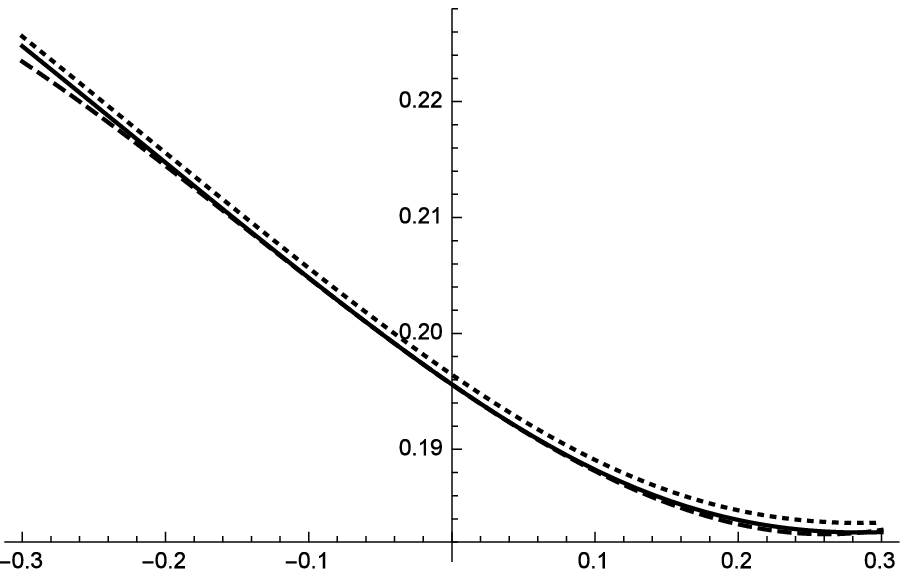}&
\includegraphics[width=.475\textwidth,height=0.19\textheight]{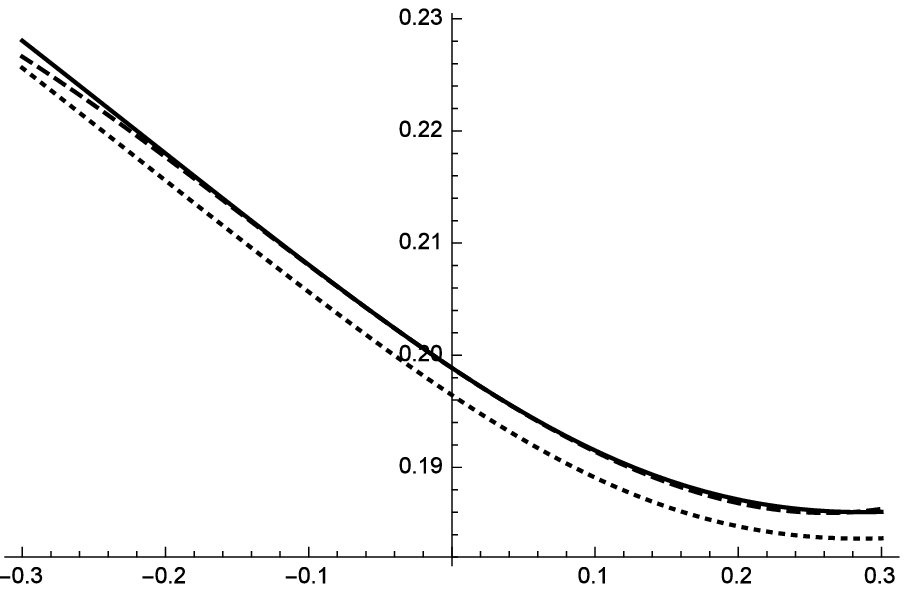}\\
$\tau=1.0$ & $\tau=1.0$ \\
\includegraphics[width=.475\textwidth,height=0.19\textheight]{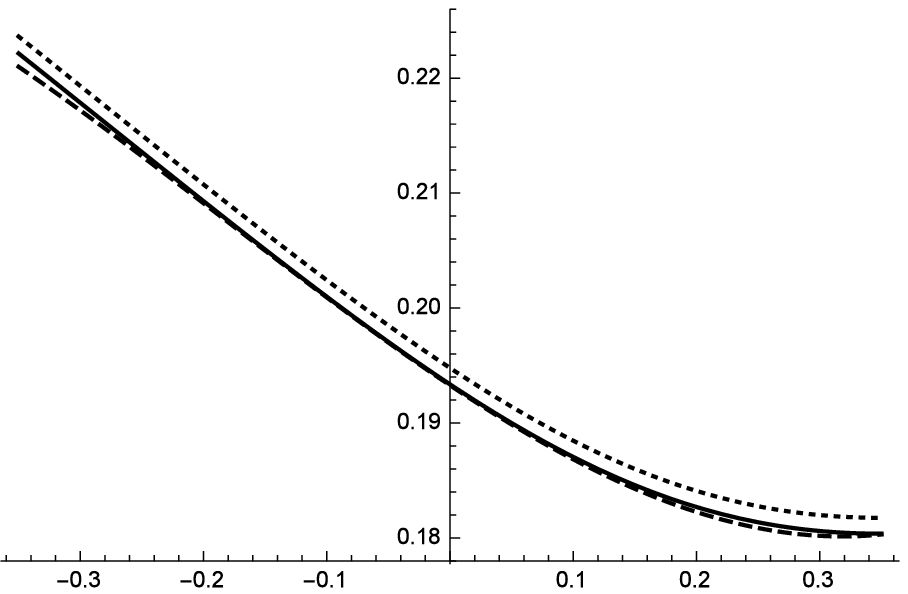}&
\includegraphics[width=.475\textwidth,height=0.19\textheight]{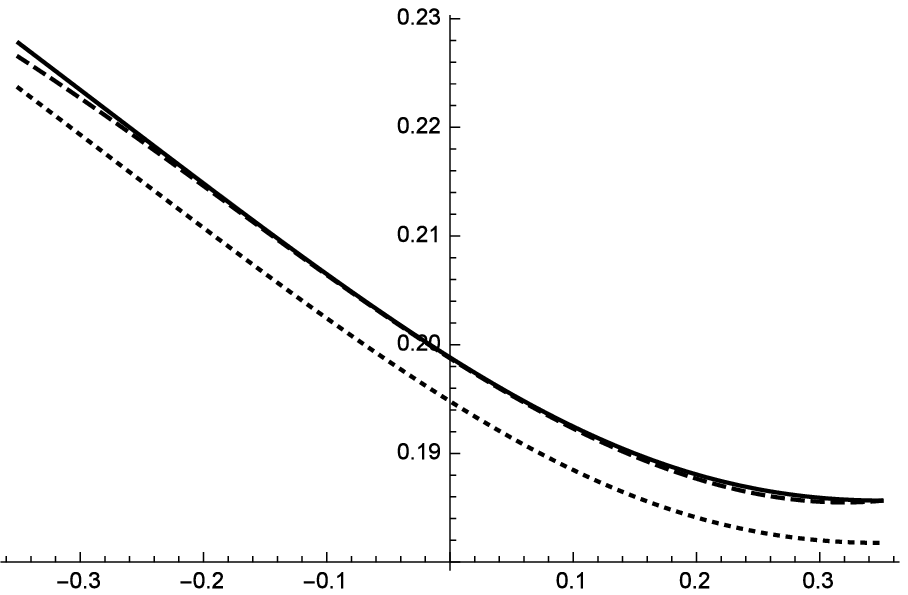}\\
$\beta= + 2$ & $\beta=-2$ \\
\end{tabular}
\caption{ \small{Exact (solid -- computed by Fourier inversion) and approximate (dashed) scaled implied
volatility $\sig_Z^{(1/\beta)}(\tau,\lam)$ under Heston model dynamics plotted as a function of log-moneyness
$\lam$.  For comparison, we also plot the exact implied volatility of the Heston model
$\sig_Z^{(1)}(\tau,\lam)=\sig_X(\tau,\lam)$ (dotted).  Parameters: $\kappa = 1.15$, $\theta=0.04$,
$\del=0.2$, $\rho=-0.4$, $y=\log \theta$.
For $\beta = \pm 2$, as $\tau$ increases, the dotted lines start  to deviate from the solid lines, but  the dashed and solids lines {remain so close they are nearly indistinguishable}.}} \label{fig:heston.iv}
\end{figure}

%
\clearpage
\begin{figure}
\centering
\begin{tabular}{  c  c }
$\tau=0.25$ & $\tau=0.25$ \\
\includegraphics[width=.475\textwidth,height=0.19\textheight]{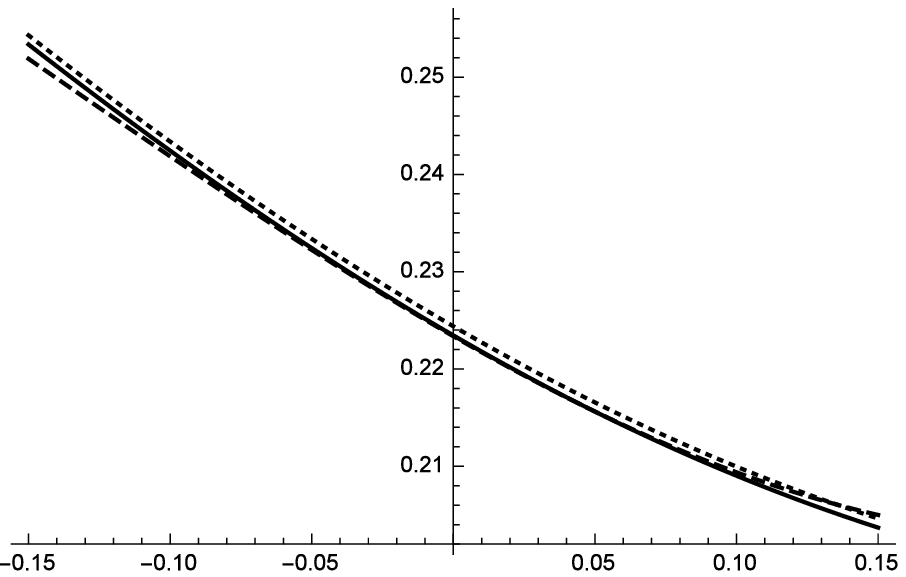}&
\includegraphics[width=.475\textwidth,height=0.19\textheight]{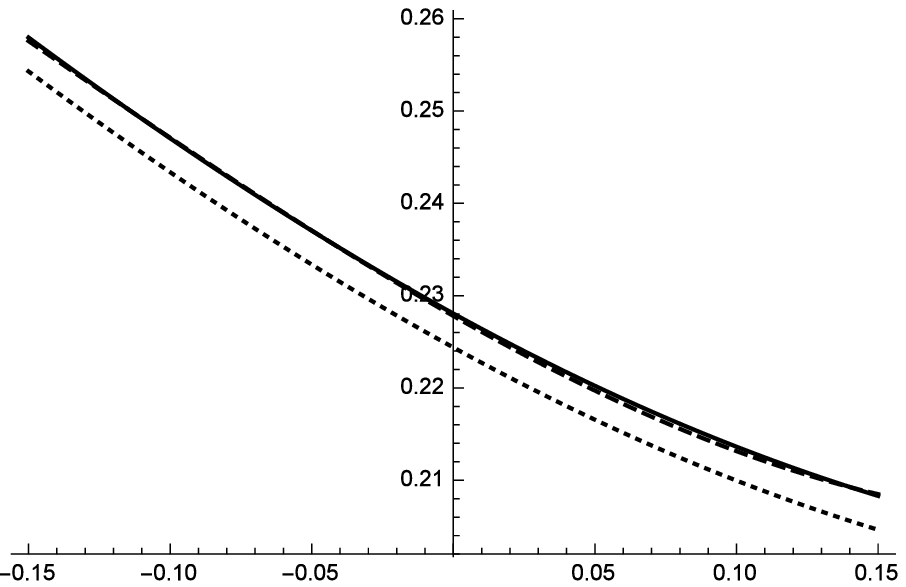}\\
$\tau=0.5$ & $\tau=0.5$ \\
\includegraphics[width=.475\textwidth,height=0.19\textheight]{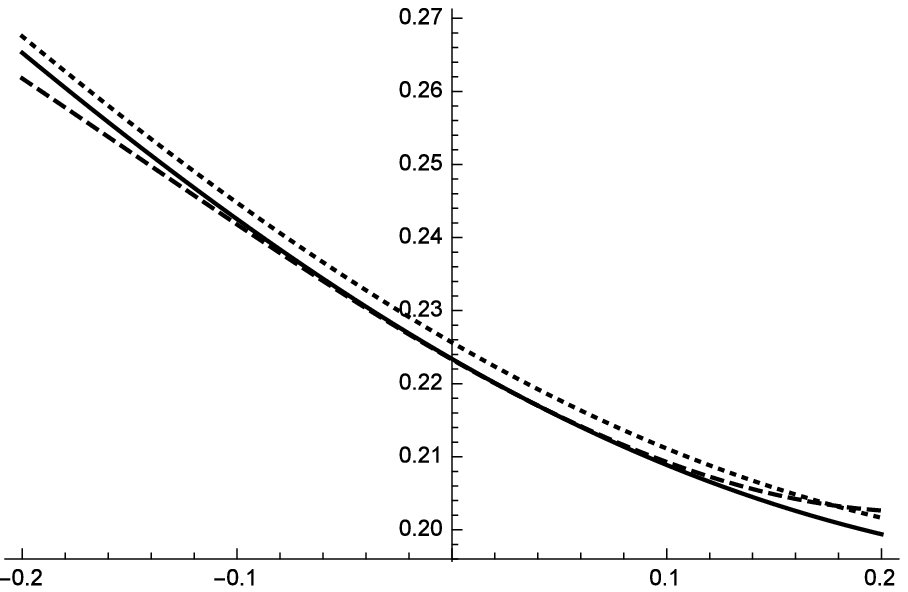}&
\includegraphics[width=.475\textwidth,height=0.19\textheight]{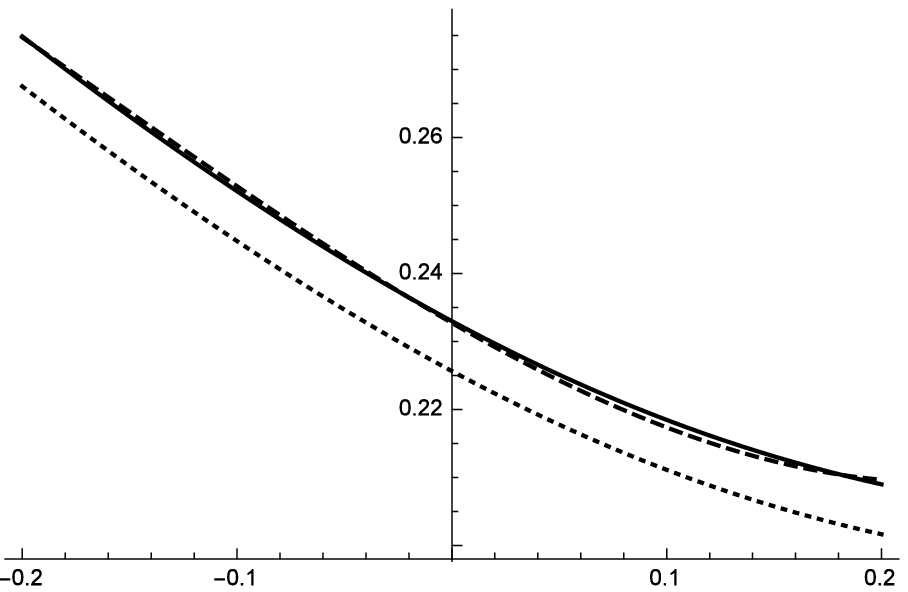}\\
$\tau=1.0$ & $\tau=1.0$ \\
\includegraphics[width=.475\textwidth,height=0.19\textheight]{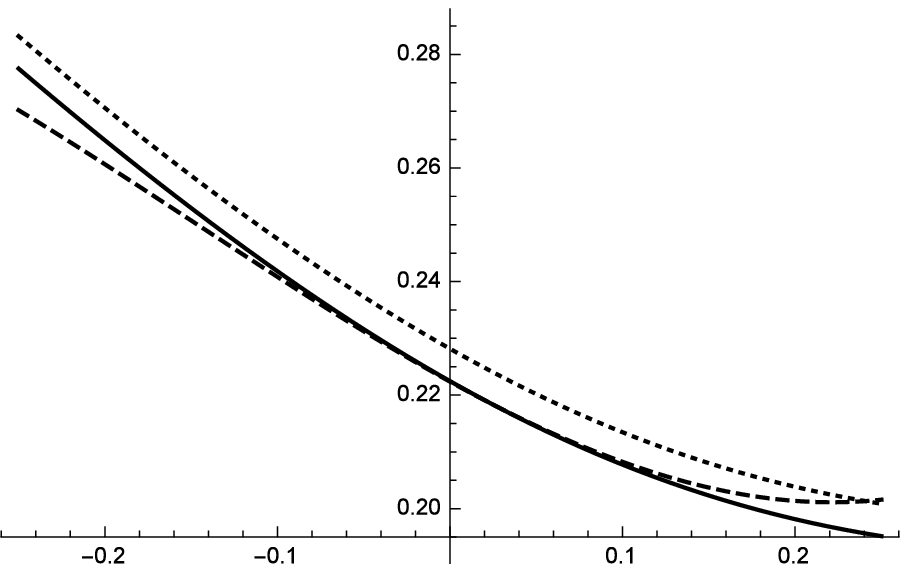}&
\includegraphics[width=.475\textwidth,height=0.19\textheight]{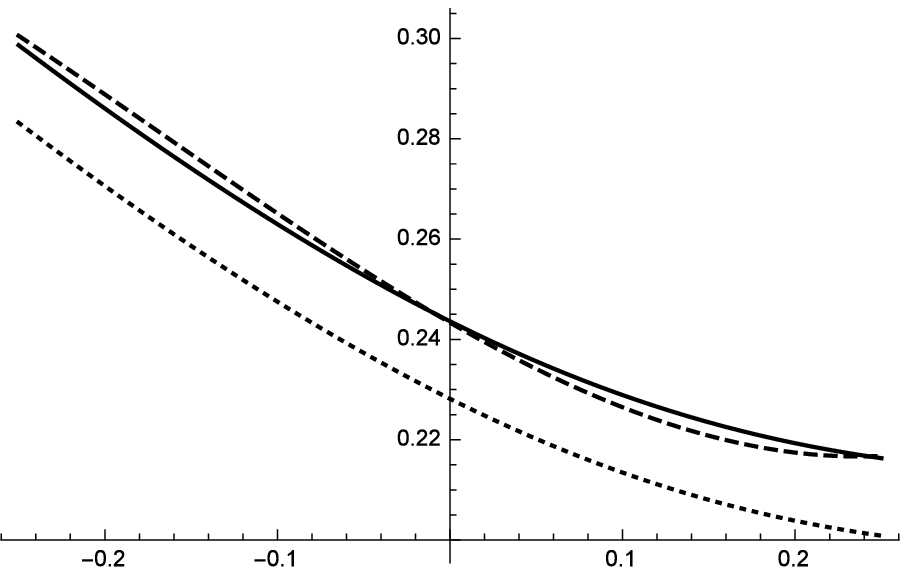}\\
$\beta= + 2$ & $\beta=-2$ \\
\end{tabular}
\caption{\small{
Exact (solid -- computed by Monte Carlo) and approximate (dashed) scaled implied volatility $\sig_Z^{(1/\beta)}(\tau,\lam)$ under SABR model dynamics plotted as a function of $\lam$.  For comparison, we also plot the exact implied volatility of the SABR model $\sig_Z^{(1)}(\tau,\lam)=\sig_X(\tau,\lam)$ (dotted).  Parameters: $\del=0.5$, $\gam=-0.5$, $\rho=0.0$ $x=0$, $y=-1.5$.
As expected, as $\tau$ increases, the solid and dotted lines diverge, while the dashed and solid lines remain close.}}
\label{fig:sabr.iv}
\end{figure}

\end{document}